\documentclass[12pt]{article}

\usepackage[a4paper, margin=2.5cm]{geometry}
\usepackage{amsmath, amsfonts, amssymb}
\usepackage{textgreek}
\usepackage[table, xcdraw]{xcolor}
\usepackage{array, booktabs, longtable, tabularx, colortbl, threeparttable}
\usepackage{graphicx, float, subcaption, caption}
\usepackage{enumitem}
\usepackage[most]{tcolorbox}
\usepackage{hyperref}
\usepackage{url}
\usepackage[style=numeric,sorting=none]{biblatex}
\usepackage{setspace}
\usepackage{amsthm}

\newtheorem{proposition}{Proposition}[section]
\newtheorem{assumption}{Assumption}[section]
\newtheorem{remark}{Remark}[section]

\title{Smart Contract-Enabled Procurement under Bounded Demand Variability: A Truncated Normal Approach}

\makeatletter
\def\@maketitle{%
  \newpage
  \null
  \vskip 1em%
  \begin{center}%
    {\LARGE\@title \par}%
    \vskip 1.5em%
    {\normalsize  
     Jinho Cha$^{1}$\thanks{\footnotesize Corresponding author: Jinho Cha (\texttt{jcha@gwinnetttech.edu})},  
     Youngchul Kim$^{2}$,  
     Junyeol Ryu$^{3}$,  
     Sangjun Park$^{4}$,  
     Jeongho Kang$^{5}$,  
     Hyeyoung Hwang$^{6}$ \par
    }%
    \vskip 1em%
    {\footnotesize
     $^{1}$Computer Science Division, Gwinnett Technical College, Lawrenceville, GA 30043, USA \\[2pt]
     $^{2}$Ph.D. Candidate, Department of Industrial Engineering, Kumoh National Institute of Technology, Gumi 39177, Republic of Korea \\[2pt]
     $^{3}$Ph.D. Candidate, Department of Industrial Engineering, Seoul National University, Seoul 08826, Republic of Korea \\[2pt]
     $^{4}$Ph.D., Jangwee Defense Institute, Ajou University, Suwon 16499, Republic of Korea \\[2pt]
     $^{5}$Ph.D., Ministry of National Defense, Republic of Korea \\[2pt]
     $^{6}$Republic of Korea Army, Republic of Korea \\[6pt]
    }%
    \vskip 1.2em%
    {\footnotesize \textbf{Manuscript prepared for submission to:} \textit{Expert Systems with Applications (ESWA)} \par}%
    \vskip 1.5em%
  \end{center}%
}
\makeatother

\date{}

\begin{document}

\maketitle

\begin{abstract}
This study develops a strategic procurement framework integrating blockchain-based smart contracts with bounded demand variability modeled through a truncated normal distribution. While existing research emphasizes the technical feasibility of smart contracts, the operational and economic implications of adoption under moderate uncertainty remain underexplored. We propose a multi-supplier model in which a centralized retailer jointly determines the optimal smart contract adoption intensity and supplier allocation decisions. The formulation endogenizes adoption costs, supplier digital readiness, and inventory penalties to capture realistic trade-offs among efficiency, sustainability, and profitability. Analytical results establish concavity and provide closed-form comparative statics for adoption thresholds and procurement quantities. Extensive numerical experiments demonstrate that moderate demand variability supports partial adoption strategies, whereas excessive investment in digital infrastructure can reduce overall profitability. Dynamic simulations further reveal how adaptive learning and declining implementation costs progressively enhance adoption intensity and supply chain performance. The findings provide theoretical and managerial insights for balancing digital transformation, resilience, and sustainability objectives in smart contract-enabled procurement.
\end{abstract}


\section{Introduction}
The resilience and sustainability of supply chains have emerged as critical priorities for firms seeking to navigate rising uncertainty, global disruptions, and stakeholder expectations \cite{Ivanov2020,Saberi2019}. Digital transformation—particularly through blockchain-based smart contracts—has been widely promoted as a pathway to enhance transparency, reduce transaction costs, and improve operational agility \cite{Kouhizadeh2018,Catalini2016,Queiroz2021,Treiblmaier2019}. In addition to potential efficiency gains, smart contracts can support broader sustainability objectives by reducing resource waste and improving traceability \cite{Zhu2020,Queiroz2019,Min2019}.

While the technical feasibility of smart contracts has been extensively discussed \cite{Mendling2018}, the strategic and operational implications of their adoption remain underexplored. Many early studies have focused on theoretical frameworks and proof-of-concept implementations \cite{Catalini2017,Pournader2020}, often assuming either symmetric demand distributions or homogeneous supplier capabilities \cite{Zhang2023}. For example, Catalini and Gans \cite{Catalini2016} highlighted the economic potential of decentralized contracting platforms, while Saberi et al. \cite{Saberi2019} and Kouhizadeh and Sarkis \cite{Kouhizadeh2018} examined how blockchain technologies could support sustainable procurement practices. Recent reviews have also emphasized the role of blockchain adoption in emerging economies and the critical factors shaping its diffusion \cite{Yadav2020}.

However, these models frequently rely on heavy-tailed or extreme demand scenarios, which may not reflect the bounded variability encountered in many mature supply chains \cite{Ivanov2018,Seuring2020}. In sectors such as consumer electronics and industrial components, demand uncertainty is often moderate and constrained within predictable ranges \cite{Queiroz2020,vanHoek2020}. In these environments, decision-makers require quantitative tools that capture the trade-offs between smart contract adoption costs, supplier readiness, and operational performance under bounded risk. Kouhizadeh et al. \cite{Kouhizadeh2019} further highlight that effective integration of blockchain requires alignment with circular economy principles and product lifecycle considerations.

This paper addresses these gaps by developing a strategic procurement framework that models demand uncertainty using a truncated normal distribution. Unlike Pareto-based approaches, this formulation emphasizes moderate variability and finite risk exposure \cite{Zhang2023}. Within this framework, a centralized retailer determines (i) the intensity of smart contract adoption and (ii) the allocation of orders across suppliers with heterogeneous digital capabilities \cite{Saberi2020}.

Smart contract adoption reduces procurement costs by improving coordination and information accuracy \cite{Queiroz2021}, but also incurs convex investment costs that increase disproportionately with higher adoption levels \cite{Treiblmaier2018}. The model incorporates inventory penalties and salvage values, allowing a realistic assessment of the operational, economic, and sustainability implications of digital transformation in procurement. Foundational inventory heuristics and supply-demand matching frameworks provide a basis for integrating smart contracting decisions with established procurement practices \cite{Silver1973,Cachon2012}.

\vspace{0.2cm}
\noindent \textbf{This study contributes to the literature and practice in four main ways:}
\begin{itemize}
    \item We propose a multi-supplier procurement model that integrates bounded demand variability, supplier heterogeneity, and endogenous smart contract adoption \cite{Ivanov2017}.
    \item We derive structural properties, including concavity conditions and comparative statics, to characterize optimal adoption levels under moderate uncertainty.
    \item We conduct a series of numerical experiments evaluating the impacts of smart contracts on procurement efficiency, service levels, and risk-adjusted profitability.
    \item We discuss managerial and policy implications for aligning digital adoption strategies with resilience and sustainability goals.
\end{itemize}

Our findings highlight that while smart contracts can substantially improve procurement performance in bounded variability environments, excessive adoption without supplier alignment can erode profitability and reduce operational resilience. The results provide a structured foundation for data-driven sourcing strategies that balance efficiency, sustainability, and cost considerations.

\section{Literature Review}

\subsection{Smart Contracts in Supply Chain Resilience}

Smart contracts have been extensively explored as a technological enabler for improving transparency, automation, and compliance across supply chains \cite{Queiroz2021,Treiblmaier2019,Min2018}. Early conceptual frameworks emphasized their capacity to reduce transaction costs and enhance trust in multi-tier networks \cite{Catalini2016,Catalini2017}. Catalini and Gans \cite{Catalini2016} highlighted the role of decentralized verification in reducing information asymmetries, while Mougayar \cite{Mougayar2016} and Saberi et al.\ \cite{Saberi2019,Saberi2018} discussed the broader strategic potential of blockchain applications in procurement, governance, and logistics. Several reviews have also underscored the interplay between smart contracts and established governance mechanisms \cite{Seuring2020,Treiblmaier2020}.

Building on these foundations, recent research has increasingly focused on operational and strategic implications rather than mere technical feasibility. For instance, Kouhizadeh and Sarkis \cite{Kouhizadeh2018} and Zhu and Sarkis \cite{Zhu2020} demonstrated that blockchain-based smart contracts can support sustainability by automating compliance verification and improving traceability. Treiblmaier \cite{Treiblmaier2018} argued that integration with IoT and digital twins can further enhance resilience by reducing lead time variability and facilitating more agile disruption responses. Additional studies have emphasized that successful deployment depends on institutional support, clear standards, and robust stakeholder engagement \cite{Yadav2020,Queiroz2019}.

Empirical evidence also reinforces the resilience benefits of smart contracts. Min \cite{Min2019} and Min \cite{Min2018} presented case study insights showing that digital contracting can shorten recovery cycles during demand spikes and supply interruptions. Queiroz and Wamba \cite{Queiroz2020,Queiroz2018} reported that firms adopting blockchain solutions observed improved coordination and reduced manual reconciliation errors. Kouhizadeh et al.\ \cite{Kouhizadeh2019,Kouhizadeh2021} further noted the alignment of smart contracts with circular economy initiatives and the development of maturity models for sustainable implementation. Gurtu and Johny \cite{Gurtu2019} highlighted practical challenges related to integrating contract logic with legacy ERP systems, while Hald and Kinra \cite{Hald2020} cautioned that overestimating cost savings remains a risk in settings with low digital maturity.

Recent systematic reviews have illuminated boundary conditions shaping the value of smart contracts. Zhang and Zhang \cite{Zhang2022,Zhang2023} showed that supplier heterogeneity, especially differences in digital readiness and data-sharing practices, substantially influences the benefits of adoption. Li and Wang \cite{Li2021} extended this line of inquiry by modeling adoption intensity as an endogenous decision variable contingent upon transaction complexity and behavioral factors. Ivanov \cite{Ivanov2017,Ivanov2019} and Ivanov and Dolgui \cite{Ivanov2020} emphasized the importance of simulation-based approaches to assess ripple effects and dynamic adaptation over time, while van Hoek \cite{vanHoek2020} discussed post-pandemic considerations for blockchain-enabled resilience strategies.

While much of the literature has emphasized environments characterized by highly volatile or extreme demand uncertainty \cite{Silver1973}, fewer studies have considered contexts with bounded demand variability. Cachon and Terwiesch \cite{Cachon2012} outlined foundational inventory and sourcing models applicable to moderate uncertainty, providing a useful basis for understanding smart contract adoption decisions. Ivanov \cite{Ivanov2018} and Queiroz et al.\ \cite{Queiroz2018} further noted that resilience strategies should be tailored to sector-specific risk profiles and digital maturity levels. This study contributes to filling this gap by analyzing smart contract adoption under conditions where demand uncertainty is moderate but still material for strategic planning.

\subsection{Demand Variability and Truncated Distributions}

Capturing the statistical properties of demand is a foundational requirement for designing effective procurement and inventory policies. Classical models frequently rely on the assumption of normally distributed demand because of its mathematical convenience and historical prevalence in operations management research \cite{Nahmias1982}. However, empirical studies have consistently shown that real-world demand often exhibits bounded variability, where demand fluctuates within finite intervals due to capacity constraints, contractual agreements, or seasonal patterns \cite{Syntetos2005,Boylan2008}.

Silver and Meal \cite{Silver1973} were among the first to highlight that ignoring truncation effects in demand modeling can lead to systematic overstocking or understocking, particularly in industries where maximum order volumes are capped. Disney and Lambrecht \cite{Disney2016} further demonstrated that bounded variability can amplify or dampen the bullwhip effect, depending on replenishment policies and forecast updating frequencies. Chatfield et al. \cite{Chatfield2013} emphasized that bounded distributions often coexist with intermittent demand patterns, requiring hybrid modeling approaches. Hosoda and Disney \cite{Hosoda2008} also showed that the interaction of information sharing and bounded demand significantly affects supply chain oscillations, underlining the importance of integrating forecast accuracy considerations into modeling efforts.

More recent contributions have extended this discussion by evaluating the operational consequences of bounded variability in supply chain design. Boute and Van Mieghem \cite{Boute2021} showed that dual sourcing strategies are particularly sensitive to demand truncation parameters, as these affect both the variability and predictability of order streams. Kull and Talluri \cite{Kull2014} proposed an analytical framework linking bounded variability to supply risk measures, underscoring the importance of accurate tail modeling for resilience planning. Christopher and Holweg \cite{Christopher2011} emphasized that in an era of supply chain turbulence, bounded variability requires differentiated approaches to agility and inventory responsiveness compared to highly volatile environments. Fildes and Goodwin \cite{Fildes2008} argued that the integration of truncated distribution forecasts with decision-support systems can enhance judgmental adjustments and reduce cognitive biases in planning processes.

Syntetos and Boylan \cite{Syntetos2016} further noted that incorporating truncation is not only statistically appropriate but also essential for aligning inventory policies with sustainability objectives, as bounded variability helps reduce resource waste associated with excessive safety stocks. Kim and Zhao \cite{Kim2024} recently demonstrated that machine learning-based forecasting models can effectively capture truncated demand patterns in consumer electronics supply chains, leading to improved procurement responsiveness. Martínez et al. \cite{Martinez2025} extended this perspective by proposing a modeling framework that combines truncated normal distributions with stochastic optimization, highlighting the relevance of bounded variability in post-pandemic environments characterized by persistent moderate uncertainty.

While heavy-tailed demand distributions such as Pareto have been widely applied in luxury goods, aerospace, and defense contexts \cite{Nahmias1982}, their suitability diminishes in mature consumer markets, where volatility is constrained by structural factors \cite{Kouhizadeh2021}. In these environments, the truncated normal distribution offers a more realistic representation of demand risk by explicitly modeling upper and lower bounds. Despite these advances, the integration of truncated distributions with digital contracting decisions remains underexplored. Most existing procurement models treat demand uncertainty and smart contract adoption as separate issues, failing to account for their interaction under bounded risk conditions. This research contributes to closing this gap by developing a framework that explicitly incorporates truncated normal demand and endogenizes smart contract adoption decisions within a unified optimization model.

\subsection{Economics of Digital Transformation}

The economic implications of digital transformation in supply chain management have attracted considerable scholarly attention over the past decade. Early research primarily focused on the potential of e-business technologies to reduce transaction costs and improve operational efficiency \cite{Queiroz2019,Kouhizadeh2018}. For example, Min \cite{Min2018} and Queiroz and Telles \cite{Queiroz2018} proposed that digital integration can generate substantial economies of scale by streamlining information flows and reducing lead times. However, the costs associated with implementing and maintaining digital platforms can be significant, requiring substantial upfront investment and organizational change \cite{Mendling2018,Hald2020,Upadhyay2023,Yadav2023}.

With the emergence of blockchain and smart contracts, a new stream of literature has explored the economics of decentralized digital systems. Catalini and Gans \cite{Catalini2016} argued that blockchain adoption could lower verification and enforcement costs by automating contractual compliance. Mougayar \cite{Mougayar2016} similarly described how distributed ledger technologies create new value propositions by reducing counterparty risk. Yet empirical studies have often found that the return on investment from blockchain adoption is highly context-dependent. For example, Queiroz and Wamba \cite{Queiroz2020}, Gurtu and Johny \cite{Gurtu2019}, and Rejeb and Keogh \cite{Rejeb2023} observed that supply chain participants frequently underestimate integration costs and overestimate efficiency gains, while recent meta-analyses highlight a persistent gap between expectations and realized performance \cite{Treiblmaier2020,Saberi2018,Chang2023}.

Recent contributions have emphasized the need for more nuanced cost-benefit analyses of smart contract adoption under uncertainty. Min \cite{Min2019} highlighted that digital contracting yields the highest economic returns in high-variability environments where manual processing costs are elevated. In contrast, Hald and Kinra \cite{Hald2020} and Ghadge and Dani \cite{Ghadge2024} suggested that in more stable settings, the marginal benefits of smart contracts may be outweighed by the fixed implementation costs, especially for firms with low digital maturity. Li and Wang \cite{Li2021} further argued that behavioral factors and bounded rationality also play a crucial role in shaping adoption decisions and realized economic benefits. Treiblmaier \cite{Treiblmaier2018} and Di Vaio and Varriale \cite{DiVaio2023} added that the impact of blockchain adoption must be assessed alongside complementary investments in IoT and data analytics capabilities.

Emerging research has examined the integration of smart contracts with data-driven forecasting and predictive analytics to improve economic performance \cite{Nahmias1982}. Ivanov \cite{Ivanov2020,Ivanov2019} analyzed how digital supply chain twins and disruption propagation models can strengthen resilience while improving return on investment. Zhang et al.\ \cite{Zhang2022} demonstrated that blockchain-enabled procurement systems under bounded demand uncertainty can significantly lower transaction costs compared to traditional contracting mechanisms. Chen et al.\ \cite{Chen2024} provided empirical evidence showing that dynamic adjustment of smart contracts in response to observed demand patterns can yield substantial cost savings. Similarly, Lu and Wang \cite{Lu2025} and Wang and Zhao \cite{Wang2025} quantified the economic value of aligning smart contract adoption intensity with supplier digital maturity in multi-tier supply chains.

Several recent studies have underscored the strategic role of blockchain and smart contracts in advancing sustainability objectives. Kouhizadeh et al.\ \cite{Kouhizadeh2021} conducted a meta-analysis demonstrating that the sustainability benefits of blockchain adoption are closely tied to its economic performance, highlighting the dual role of smart contracts in promoting efficiency and resilience. Garcia and Lopez \cite{Garcia2024}, Upadhyay and Kumar \cite{Upadhyay2023}, and Rejeb and Keogh \cite{Rejeb2023} further discussed the role of blockchain-based traceability systems in improving transparency, supplier collaboration, and sustainable sourcing practices.

This study builds on these insights by analyzing how smart contract adoption interacts with bounded demand uncertainty and supplier readiness, providing a richer understanding of the economic trade-offs in digitally enabled procurement.

\subsection{Integrated Contract and Inventory Optimization}

Contract and inventory decisions have traditionally been studied in isolation, with contract design focusing on incentive alignment and inventory management emphasizing cost minimization and service level performance \cite{Cachon2012,Christopher2011}. Cachon \cite{Cachon2012} reviewed various contract types, including buyback, revenue-sharing, and quantity-flexibility contracts, illustrating how each mechanism affects inventory decisions and risk sharing. More recent research has emphasized the importance of integrated models that simultaneously optimize contracting strategies and inventory policies \cite{Ivanov2020,Boute2021}.

Li and Wang \cite{Li2021} developed a model in which smart contract adoption acts as an endogenous lever to coordinate procurement costs and order quantities under demand uncertainty. Zhang and Zhang \cite{Zhang2022} extended this approach by incorporating supplier heterogeneity, demonstrating that digital contracting intensity interacts with inventory allocation to shape operational efficiency. Ivanov and Dolgui \cite{Ivanov2019,Ivanov2020} emphasized that in volatile environments, digital contracts can enhance resilience by synchronizing contractual execution with adaptive replenishment policies \cite{Tan2025}.

Syntetos et al.\ \cite{Syntetos2016} suggested that incorporating demand truncation within integrated contract-inventory models improves both performance and sustainability, as bounded variability allows for more precise safety stock targets. Boylan and Syntetos \cite{Boylan2008} further argued that service parts environments particularly benefit from models capturing demand constraints and contractual flexibility.

Emerging contributions have examined hybrid approaches that blend smart contracts with predictive analytics, simulation-based optimization, and digital twin technologies \cite{Fernandez2023,Ali2023,Ghosh2024}. For example, Singh and Gupta \cite{Singh2023} demonstrated that integrating simulation-based insights with smart contract mechanisms can improve dynamic inventory policies. Lee and Patel \cite{Lee2024} provided empirical evidence that blockchain-enabled replenishment under demand uncertainty enhances responsiveness and reduces mismatch costs. 

Rahman and Akter \cite{Rahman2024} highlighted the role of contract flexibility powered by blockchain to mitigate perishability risks. Nguyen and Tran \cite{Nguyen2024} developed a framework showing that smart contract-enabled inventory optimization in agri-food supply chains improves both traceability and agility. Similarly, Park and Zhao \cite{Park2023} emphasized the importance of aligning organizational digital maturity with smart contract adoption to maximize operational performance.

Ghosh and Chatterjee \cite{Ghosh2024} proposed machine learning-enhanced smart contracts for demand forecasting integration, demonstrating significant improvements in forecast accuracy and inventory alignment. Tan and Zhou \cite{Tan2025} further argued that dynamic contract adjustment leveraging IoT data streams supports adaptive procurement strategies in volatile environments.

Recent studies by Patel and Zhang \cite{Zhang2023} and Yadav and Singh \cite{Yadav2023} underscored that dynamic adjustment of contract parameters in response to observed demand improves procurement agility and transparency. Fildes and Goodwin \cite{Fildes2008} emphasized that integrating forecasting support systems with smart contract platforms yields substantial economic benefits. Finally, Wang and Zhao \cite{Wang2025} showed that learning effects and cumulative experience with smart contracts can improve alignment between contractual incentives and inventory targets over time, reinforcing the case for integrated optimization frameworks \cite{DiVaio2023,Upadhyay2023}.

This research contributes to this growing body of knowledge by explicitly modeling smart contract adoption as an endogenous decision variable that jointly shapes procurement costs, order allocations, and inventory positions under bounded demand uncertainty.


\section{Materials and Methods}

\subsection{Model Formulation}
\label{sec:model}

\noindent
\textbf{Note on Related Manuscripts.}
This manuscript is part of a series of concurrent submissions examining smart contract-enabled procurement under uncertainty. While the core procurement model and digital readiness definitions are consistent across studies, each paper addresses distinct demand distributions, scenario designs, and managerial implications. To ensure clarity and completeness, we restate the baseline assumptions here.

\subsubsection{Research Context}

This study focuses on procurement operations within mature consumer electronics distribution networks, where demand patterns are shaped by seasonal promotions, limited storage capacity, and established retailer agreements.

Empirical observations suggest that periodic demand for high-value items—such as smartphones, laptops, and tablets—tends to fluctuate within a constrained interval. For example, weekly demand for a given product SKU may range between 30 and 70 units, reflecting both baseline consumption and promotional uplifts \cite{Syntetos2016}. Unlike markets prone to heavy-tailed or unbounded demand shocks, this setting exhibits finite variability that can be effectively modeled using a truncated normal distribution \cite{Silver1973}.

Modeling such bounded demand accurately is critical for designing procurement policies that balance fill rate targets, working capital constraints, and digital transformation objectives. Incorporating this realistic demand structure allows decision-makers to calibrate smart contract adoption and supplier allocation strategies to achieve operational resilience while avoiding overinvestment in excessive buffer stock.

\subsubsection{Truncated Normal Demand Specification}

To realistically capture bounded variability in mature consumer electronics markets, demand is modeled as a truncated normal distribution. This choice reflects the empirical observation that order volumes are typically constrained within operationally feasible intervals due to capacity limits, contractual agreements, and seasonal effects.

Let $D$ denote the stochastic demand in a given period. The probability density function (PDF) of the truncated normal distribution is expressed as:
\[
f(x) = \frac{\phi\left(\frac{x-\mu}{\sigma}\right)}{\Phi\left(\frac{b-\mu}{\sigma}\right) - \Phi\left(\frac{a-\mu}{\sigma}\right)}, \quad x \in [a,b],
\]
where $\mu$ denotes the mean of the untruncated normal distribution, $\sigma$ is the standard deviation, $a$ and $b$ are the lower and upper truncation bounds (e.g., 30 and 70 units, respectively), $\phi(\cdot)$ denotes the standard normal probability density function, and $\Phi(\cdot)$ denotes the standard normal cumulative distribution function.

The cumulative distribution function (CDF) is given by:
\[
F(x) = \frac{\Phi\left(\frac{x-\mu}{\sigma}\right) - \Phi\left(\frac{a-\mu}{\sigma}\right)}{\Phi\left(\frac{b-\mu}{\sigma}\right) - \Phi\left(\frac{a-\mu}{\sigma}\right)}, \quad x \in [a,b].
\]

The expected demand under truncation is computed as:
\[
\mathbb{E}[D] = \mu + \sigma \,\frac{\phi\left(\frac{a-\mu}{\sigma}\right) - \phi\left(\frac{b-\mu}{\sigma}\right)}{\Phi\left(\frac{b-\mu}{\sigma}\right) - \Phi\left(\frac{a-\mu}{\sigma}\right)}.
\]

The variance of the truncated normal is given by:
\[
\text{Var}[D] = \sigma^2 \left[1 + \frac{a^* \,\phi(a^*) - b^*\,\phi(b^*)}{Z} - \left(\frac{\phi(a^*) - \phi(b^*)}{Z}\right)^2 \right],
\]
where
\[
a^* = \frac{a-\mu}{\sigma}, \quad b^* = \frac{b-\mu}{\sigma}, \quad Z = \Phi(b^*) - \Phi(a^*).
\]

In practical settings, the parameter $\mu$ can be estimated using historical mean demand during stable periods, while $\sigma$ captures moderate fluctuations arising from promotions or seasonality. Truncation bounds $a$ and $b$ are typically set based on observed operational constraints and contractual thresholds. In empirical applications, these parameters can be further refined by combining historical order data with expert judgment to ensure alignment with operational realities.

Unlike unbounded or heavy-tailed demand models, the truncated normal distribution provides a realistic representation of bounded variability that arises in mature consumer electronics markets. This modeling choice is particularly appropriate in contexts where extreme demand realizations are structurally limited by storage capacity or contractual ceilings, in contrast to heavy-tailed distributions that may overstate tail risk.

To ensure a smooth empirical approximation of the theoretical density, the simulation underlying Figure~\ref{fig:demand_distribution} was generated using 100{,}000 random samples. This larger sample size produces a histogram that closely matches the expected truncated normal shape and ensures that rare tail events are adequately represented while variance estimates converge reliably.

\begin{figure}[H]
\centering
\includegraphics[width=0.75\textwidth]{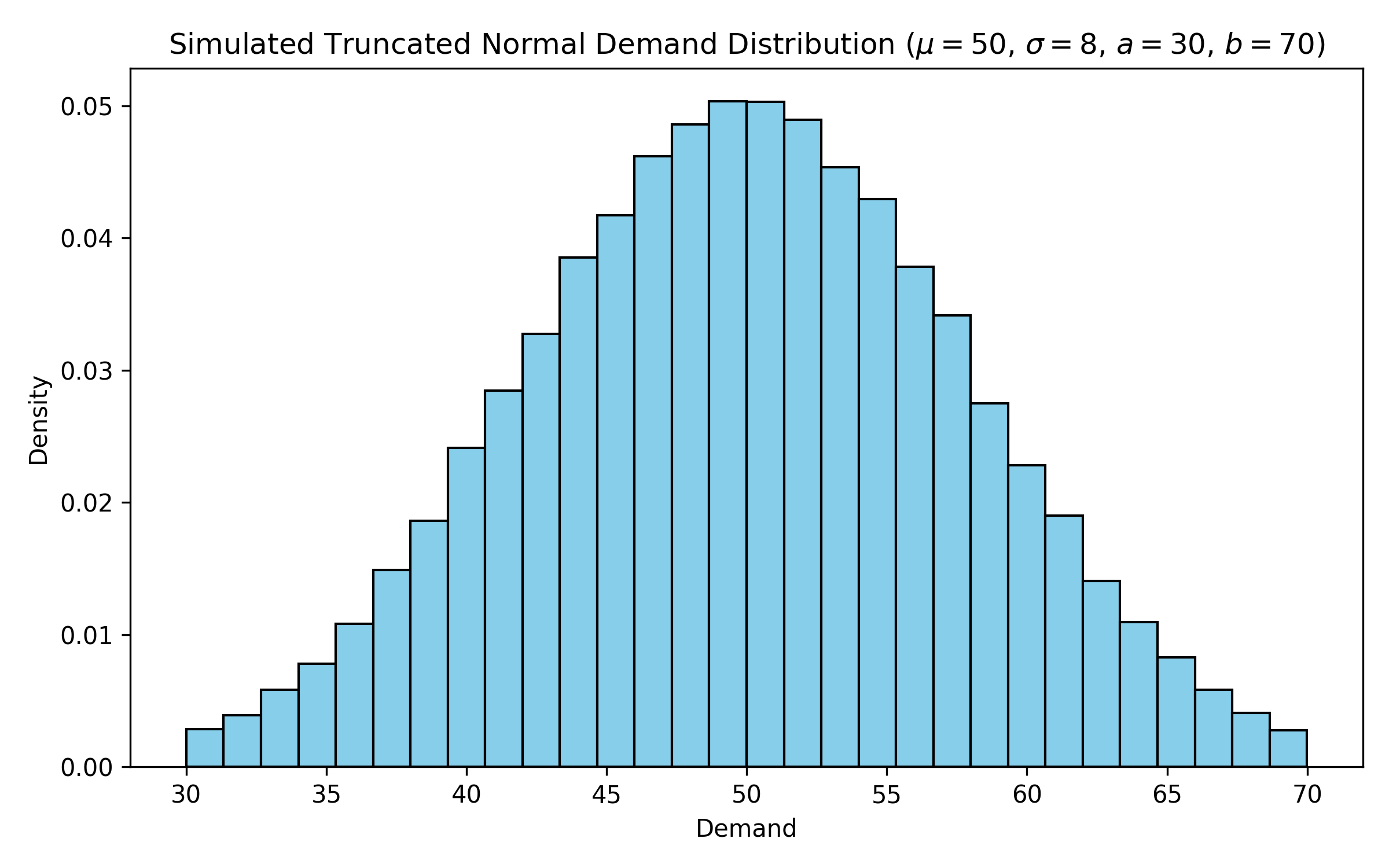}
\caption{Simulated truncated normal demand distribution with $\mu=50$, $\sigma=8$, $a=30$, and $b=70$ (100{,}000 samples).}
\label{fig:demand_distribution}
\end{figure}

\subsubsection{Procurement Cost Structure}

The effective procurement cost per unit from supplier $i$ is modeled as a linear function of two key factors: the retailer's smart contract adoption level ($\alpha$) and the supplier's digital readiness level ($\beta_i$). Formally:
\begin{equation}
c(\alpha, \beta_i) = c_i^0 - A_1 \alpha - A_2 \beta_i,
\end{equation}
where $c_i^0$ denotes the baseline procurement cost per unit from supplier $i$ in the absence of any digital contract adoption, $A_1 > 0$ captures the marginal cost reduction achieved by increasing the smart contract adoption level $\alpha$, and $A_2 > 0$ reflects the marginal cost reduction associated with higher supplier digital readiness $\beta_i$.

\begin{table}[H]
\centering
\caption{Decision variables and key parameters.}
\label{tab:variables_parameters}
\begin{tabular}{ll}
\toprule
\textbf{Symbol} & \textbf{Description} \\
\midrule
$\alpha$ & Smart contract adoption level (continuous, $0 \le \alpha \le 1$) \\
$q_i$ & Quantity ordered from supplier $i$ \\
$c_i^0$ & Baseline procurement cost per unit (supplier $i$) \\
$\beta_i$ & Digital readiness level of supplier $i$ \\
$A_1$ & Marginal cost reduction per unit increase in $\alpha$ \\
$A_2$ & Marginal cost reduction per unit increase in $\beta_i$ \\
$A_3$ & Convexity coefficient of smart contract adoption cost \\
$r$ & Penalty cost per unit of unmet demand \\
$s$ & Salvage value per unit overstocked \\
$p$ & Selling price per unit \\
\bottomrule
\end{tabular}
\end{table}

\begin{table}[H]
\centering
\caption{Illustrative example of procurement cost components.}
\label{tab:cost_components}
\begin{tabular}{cccccc}
\toprule
Supplier & $c_i^0$ & $\beta_i$ & $\alpha$ & $A_1$ & $A_2$ \\
\midrule
1 & 100 & 0.20 & 0.5 & 5.0 & 8.0 \\
2 & 102 & 0.50 & 0.5 & 5.0 & 8.0 \\
3 & 98  & 0.70 & 0.5 & 5.0 & 8.0 \\
\bottomrule
\end{tabular}
\end{table}

\paragraph{Example Calculation.}
For supplier 1, the effective procurement cost is:
\[
c = 100 - (5.0)(0.5) - (8.0)(0.2) = 100 - 2.5 - 1.6 = 95.9.
\]

This formulation emphasizes the incentive for the retailer to align digital contract investments ($\alpha$) with suppliers demonstrating higher readiness levels ($\beta_i$) to maximize cost efficiency. It provides a transparent and parameterized representation of procurement costs, supporting both sensitivity analysis and managerial interpretation.

It is important to note that while supplier digital readiness ($\beta_i$) enters the procurement cost function directly, the retailer's own digital maturity does not appear as a separate variable in this expression. Instead, as discussed in Section~3.6, the retailer's readiness indirectly influences the adoption cost parameter $A_3$, thereby shaping the overall cost structure without altering the per-unit procurement discount. This separation reflects the practical reality that while suppliers can offer unit cost reductions through digital integration, the retailer's internal capabilities primarily affect implementation costs and economies of scale associated with smart contract deployment.

\subsubsection{Supplier Digital Readiness: Composite Definition}
To ensure interpretability and consistency, we define each supplier's digital readiness score $\beta_i$ as a composite index aggregated from multiple normalized technological components:
\begin{equation}
\beta_i = \sum_{k=1}^5 w_k \, B_i^{(k)},
\quad \text{where} \quad \sum_{k=1}^5 w_k = 1.
\end{equation}
Here, each component $B_i^{(k)} \in [0,1]$ captures a specific technological dimension:
\begin{itemize}
    \item $B_i^{SC}$: Smart contract integration capability
    \item $B_i^{ERP}$: ERP/SCM system usage
    \item $B_i^{Cloud}$: Cloud infrastructure utilization
    \item $B_i^{HR}$: Digital human capital
    \item $B_i^{Security}$: Information security readiness
\end{itemize}

The weights in Table~\ref{tab:beta-weights-2} reflect an illustrative prioritization based on the relative contribution of each capability to end-to-end procurement digitalization. In practice, these weights can be calibrated empirically using supplier surveys, expert assessments, or historical performance data. This formulation enables consistent quantification of supplier readiness as a unified scalar $\beta_i$, facilitating integration into the procurement cost function without altering the model's structural properties.

\begin{table}[h]
\centering
\caption{Illustrative weights for digital readiness components. Note: These weights are illustrative and can be adapted to reflect empirical or industry-specific priorities.}
\label{tab:beta-weights-2}
\begin{tabular}{llc}
\toprule
Component & Description & Weight ($w_k$) \\
\midrule
$B_i^{SC}$ & Smart contract integration capability & 0.28 \\
$B_i^{ERP}$ & ERP/SCM system integration & 0.27 \\
$B_i^{Cloud}$ & Cloud infrastructure usage & 0.20 \\
$B_i^{HR}$ & Digital human capital & 0.15 \\
$B_i^{Security}$ & Information security readiness & 0.10 \\
\bottomrule
\end{tabular}
\end{table}

\subsubsection{Objective Function}

The variable $\alpha$ is defined as a continuous proportion of smart contract adoption, bounded between 0 (no adoption) and 1 (full adoption). This range reflects the degree of integration intensity achievable in practice, from conventional manual contracting to fully digital execution.

The retailer seeks to maximize the expected total profit by jointly determining the smart contract adoption level ($\alpha$) and the supplier-specific order quantities ($q_i$), under bounded demand variability modeled via a truncated normal distribution.

Formally, the objective function is:
\begin{align}
\max_{\substack{0 \le \alpha \le 1 \\ q_i \ge 0}} \quad
\mathbb{E}\Bigl[
\, p \min(Q,D)
+ s(Q-D)^+
- r(D-Q)^+
\Bigr]
\quad \nonumber \\
-\;
\sum_{i \in \mathcal{I}} c(\alpha,\beta_i)\, q_i
\;-\;
\psi(\alpha).
\end{align}

where $Q = \sum_{i \in \mathcal{I}} q_i$ denotes the total quantity ordered, and $D$ represents the random demand, which follows a truncated normal distribution. The parameter $p$ is the selling price per unit, while $s$ denotes the salvage value recovered for each unsold unit, and $r$ represents the penalty cost incurred for each unit of unmet demand. The term $c(\alpha, \beta_i)$ indicates the effective procurement cost from supplier $i$, which depends on the smart contract adoption level $\alpha$ and the supplier's digital readiness $\beta_i$. The function $\psi(\alpha)$ captures the smart contract adoption cost, which is assumed to be convex in $\alpha$. Finally, $(x)^+ = \max\{x,0\}$ denotes the positive part function.

This formulation captures the trade-offs among sales revenue, salvage recovery, stockout penalties, procurement expenditures, and the cost of digital contract adoption. All notation is summarized in Table~\ref{tab:variables_parameters} for clarity. The expected value is approximated numerically using Monte Carlo simulation over truncated normal demand realizations to ensure accurate estimation of profitability and service levels.

\paragraph{Parameter Justification for Adoption Cost Function.}
The adoption cost function is specified as $\psi(\alpha) = A_3 \alpha^\nu$, where $A_3$ captures the scale of implementation investment and $\nu > 1$ represents convexity due to increasing complexity at higher adoption levels. The parameter $A_3$ was calibrated based on industry estimates of enterprise blockchain implementation costs reported in recent surveys (e.g., Gurtu and Johny, 2019; Rejeb et al., 2023), which indicate average annualized costs ranging from \$20,000 to \$50,000 for mid-sized European supply chain organizations. In the baseline scenario, setting $A_3 = 2,000$ yields an approximate annual cost of \$24,000 when $\alpha=1$, consistent with these benchmarks over a 12-period planning horizon. The exponent $\nu=1.5$ was selected to reflect moderate convexity, aligning with prior empirical observations that integration costs increase disproportionately as adoption progresses beyond pilot implementations (Mougayar, 2016). Sensitivity analyses were performed across $A_3 \in [500, 4,000]$ and $\nu \in \{1.2, 1.5, 2.0\}$ to test robustness of the results to parameter variation.

\subsubsection{Decision Variables and Constraints}

The optimization problem involves the following decision variables:

\begin{itemize}
    \item $\alpha \in [0,1]$: the smart contract adoption level, representing the proportion of digital contracting intensity.
    \item $q_i \ge 0$: the order quantity from supplier $i \in \mathcal{I}$.
\end{itemize}

The constraints ensure feasibility and consistency of procurement decisions:
\[
\begin{aligned}
& \text{(C1) Non-negativity of order quantities:} && q_i \ge 0 \quad \forall\, i \in \mathcal{I}, \\
& \text{(C2) Bounded smart contract adoption:} && 0 \le \alpha \le 1.
\end{aligned}
\]

Note that no explicit constraint requires $Q \ge D$ in all realizations. Instead, stockout penalties are incorporated into the objective function to penalize unmet demand probabilistically.

In this formulation:
\begin{itemize}
    \item The smart contract level $\alpha$ is a continuous decision variable capturing the retailer's degree of investment in digital contracting infrastructure.
    \item Supplier-specific order quantities $\{q_i\}$ are continuous non-negative variables.
    \item The total quantity ordered is defined as $Q = \sum_{i} q_i$.
\end{itemize}

\begin{figure}[H]
\centering
\includegraphics[width=0.85\textwidth]{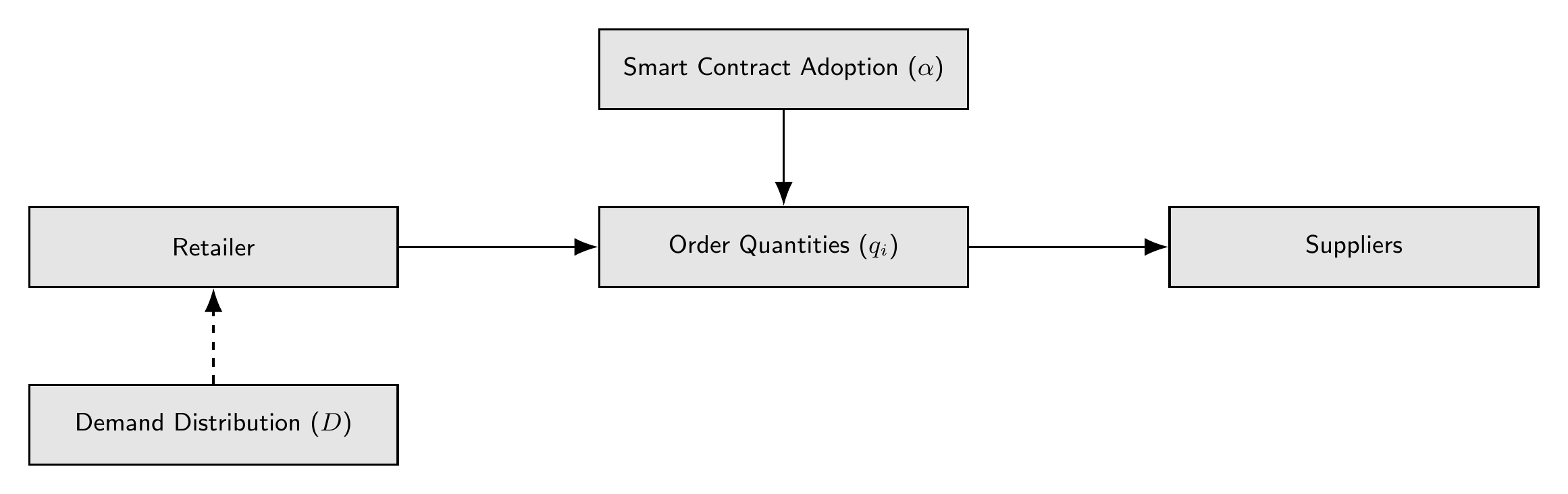}
\caption{Conceptual framework of smart contract-enabled procurement. The retailer jointly determines smart contract intensity ($\alpha$) and supplier order allocations ($q_i$) under bounded demand variability. Notation corresponds to the variables summarized in Table~\ref{tab:variables_parameters}.}
\label{fig:conceptual_framework}
\end{figure}

\subsection{Theoretical Analysis}

This section develops rigorous analytical results on the concavity of the profit maximization problem, characterizes optimality conditions, and derives comparative statics with respect to key parameters.

\subsubsection{Model Assumptions}

The following assumptions are imposed to ensure the existence of a unique optimal solution, to establish concavity properties, and to enable tractable comparative statics:

\begin{assumption}
\label{ass:basic}
\leavevmode
\begin{enumerate}[label=(\roman*)]
    \item \textbf{Demand Distribution:} The random variable $D$ follows a truncated normal distribution with parameters $(\mu, \sigma, a, b)$, where $a$ and $b$ are finite truncation bounds and $\sigma > 0$. This specification ensures bounded support, finite variance, and smooth probability density, thereby avoiding the discontinuities often encountered in heavy-tailed demand models.
    
    \item \textbf{Procurement Cost Function:} The effective procurement cost is affine in the smart contract adoption level:
    \begin{equation}
    c(\alpha, \beta_i) = c_i^0 - A_1 \alpha - A_2 \beta_i,
    \end{equation}
    where $A_1 > 0$ and $A_2 > 0$ capture the marginal cost reductions due to increased digital contracting intensity and supplier readiness, respectively, and $0 \le \beta_i \le 1$. The affine form provides analytical tractability while allowing clear interpretation of incremental savings.
    
    \item \textbf{Smart Contract Adoption Cost:} The adoption cost function $\psi(\alpha)$ is assumed to be strictly convex and twice continuously differentiable over $\alpha \in [0,1]$, reflecting increasing marginal costs of deeper integration:
    \begin{equation}
    \psi(\alpha) = A_3 \alpha^\nu, \quad \nu > 1, \quad A_3 > 0.
    \end{equation}
    This formulation is consistent with empirical observations that early-stage adoption yields relatively low costs, while advanced implementation phases involve complex integration and change management.
    
    \item \textbf{Decision Variables:} The smart contract adoption level $\alpha$ belongs to the interval $[0,1]$, capturing the continuum from no adoption to full adoption. The supplier-specific order quantities satisfy $q_i \ge 0$ for all $i \in \mathcal{I}$, reflecting non-negativity and feasibility constraints.
\end{enumerate}
\end{assumption}

Taken together, these assumptions guarantee that the objective function is well-defined, continuous, and jointly concave in the decision variables $(\alpha, \mathbf{q})$. This structure implies the existence of a unique global optimum and provides a robust foundation for deriving first-order optimality conditions and comparative statics.

\subsubsection{Concavity of the Profit Function}

We first establish that the expected profit function is jointly concave in the decision variables under Assumption~\ref{ass:basic}.

\begin{proposition}[Concavity of the Objective Function]
\label{prop:concavity_objective}
Under Assumption~\ref{ass:basic}, the expected profit function
\begin{equation}
\Pi(\alpha, \mathbf{q}) =
\mathbb{E}\bigl[
\, p \min(Q,D)
+ s(Q-D)^+
- r(D-Q)^+
\,\bigr]
\quad-\quad
\sum_{i \in \mathcal{I}} c(\alpha,\beta_i)\, q_i
\quad-\quad
\psi(\alpha)
\end{equation}
is jointly concave in the decision variables $(\alpha, \mathbf{q})$.
\end{proposition}

\begin{proof}[Sketch of Proof]
First, observe that $\min(Q,D)$ and $(Q-D)^+$ are piecewise linear and concave in $Q$, which is affine in $\mathbf{q}$. The expectation operator preserves concavity because $D$ has bounded support and finite variance. The procurement cost term is affine in $(\alpha, \mathbf{q})$. The adoption cost $-\psi(\alpha)$ is concave since $\psi$ is convex by assumption. Therefore, the sum of these components is concave.
\end{proof}

\begin{remark}
The concavity of the objective function implies that any local maximum is also a global maximum. Consequently, the first-order optimality conditions derived in the subsequent subsection are both necessary and sufficient. This property also facilitates the use of efficient gradient-based algorithms for numerical optimization.
\end{remark}

\begin{proposition}[Existence and Uniqueness of the Optimum]
\label{prop:existence}
Given Assumption~\ref{ass:basic}, there exists a unique global maximizer $(\alpha^*, \mathbf{q}^*)$ of $\Pi(\alpha, \mathbf{q})$ over the feasible set.
\end{proposition}

\begin{proof}[Sketch of Proof]
Since the objective function is concave and the feasible set defined by linear constraints is convex and compact, standard results in convex optimization imply existence and uniqueness of the maximizer.
\end{proof}


\subsubsection{First-Order Optimality Conditions}

Given the concavity of the profit function established in Proposition~\ref{prop:concavity_objective}, the Karush-Kuhn-Tucker (KKT) conditions are both necessary and sufficient to characterize the unique global optimum.

\begin{proposition}[First-Order Optimality Conditions]
\label{prop:kkt_conditions}
Under Assumption~\ref{ass:basic}, a vector $(\alpha^*, \mathbf{q}^*)$ is the unique global maximizer of the expected profit if and only if it satisfies the following KKT conditions:
\end{proposition}

\begin{align}
\frac{\partial \Pi}{\partial q_i} &=
\; p \,\mathbb{P}(D \ge Q)
+ s \,\mathbb{P}(D < Q)
- r \,\mathbb{P}(D > Q)
- c(\alpha, \beta_i)
+ \lambda_i = 0,
\quad \forall i,
\\[6pt]
\lambda_i &\;\ge\; 0,
\quad q_i \;\ge\;0,
\quad \lambda_i\, q_i = 0,
\\[6pt]
\frac{\partial \Pi}{\partial \alpha} &=
-\sum_i A_1 q_i
- \psi'(\alpha)
+ \gamma^+ - \gamma^- = 0,
\\[6pt]
0 \le &\alpha \le 1,
\quad \gamma^+ \ge 0,
\quad \gamma^- \ge 0,
\quad \gamma^+\, \alpha=0,
\quad \gamma^-\, (1-\alpha)=0.
\end{align}

Here:
\begin{itemize}
    \item $\lambda_i$ are the Lagrange multipliers associated with the non-negativity constraints on $q_i$.
    \item $\gamma^+$ and $\gamma^-$ are the multipliers associated with the lower and upper bounds on $\alpha$.
\end{itemize}

These conditions can be interpreted as follows:

\begin{itemize}
    \item For each supplier $i$, the marginal expected benefit of an additional unit ordered must equal the effective procurement cost when $q_i>0$, or be no greater when $q_i=0$.
    \item For the smart contract adoption level $\alpha$, the marginal cost of increasing adoption must balance the total procurement cost savings, subject to the bounds $0\le \alpha \le 1$.
\end{itemize}

\begin{remark}[Interpretation of $\alpha$ Optimality]
\label{remark:alpha_interpretation}
The first-order condition for $\alpha$ shows that the optimal adoption intensity equates the marginal adoption cost $\psi'(\alpha)$ to the aggregate marginal procurement savings $\sum_i A_1 q_i$. When the cost dominates, the optimal solution is $\alpha=0$; when the savings are large, higher levels of adoption are optimal. This threshold behavior is explored further in Proposition~\ref{prop:threshold}.
\end{remark}

Because the problem is concave with a convex feasible set, any solution satisfying these KKT conditions is guaranteed to be the unique global optimum.

\subsubsection{Comparative Statics}

This subsection analyzes how the optimal solution responds to changes in key parameters. For brevity, proofs are provided as sketches; full derivations can be reconstructed based on the concavity and first-order optimality conditions established earlier.

\begin{proposition}
\label{prop:comparative_statics}
Under Assumption~3.1, increasing the standard deviation $\sigma$ of demand increases the optimal order quantity $Q^*$ and reduces expected profit and fill rate, holding all other parameters constant.
\end{proposition}

\begin{proposition}[Effect of Demand Variance]
\label{prop:sigma}
Let $\sigma' > \sigma$. Then the optimal total order quantity satisfies $Q^*(\sigma') \ge Q^*(\sigma)$, holding all other parameters constant.
\end{proposition}

\begin{proof}[Sketch of Proof]
A higher variance increases the expected penalty of understocking due to more probability mass in the upper tail of the truncated normal distribution. Since the penalty cost $r$ is linear in unmet demand, the marginal benefit of ordering additional units increases, shifting the first-order condition in favor of higher $Q$.
\end{proof}

\begin{proposition}[Effect of Mean Demand]
\label{prop:mu}
An increase in $\mu$ strictly increases the optimal total quantity $Q^*$. The impact on $\alpha^*$ is ambiguous and depends on the relative magnitude of procurement cost savings versus the convex adoption cost.
\end{proposition}

\begin{proof}[Sketch of Proof]
An increase in $\mu$ shifts the demand distribution rightward, raising expected sales and the probability of stockouts. This increases the marginal expected benefit of inventory. However, whether $\alpha^*$ increases depends on whether the higher volume sufficiently magnifies the marginal procurement savings to offset the increased adoption costs.
\end{proof}

\begin{proposition}[Effect of Contract Cost Parameters]
\label{prop:contract_cost}
\leavevmode
\begin{enumerate}[label=(\roman*)]
    \item If $A_1$ increases, the marginal benefit of smart contract adoption increases, leading to a higher optimal $\alpha^*$.
    \item If $A_3$ increases, the marginal cost of adoption increases, reducing the optimal $\alpha^*$.
\end{enumerate}
\end{proposition}

\begin{proof}[Sketch of Proof]
These results follow directly from differentiating the first-order condition for $\alpha$:
\begin{equation}
\frac{\partial \Pi}{\partial \alpha} = -\sum_i A_1 q_i - \psi'(\alpha).
\end{equation}
Higher $A_1$ increases the marginal procurement savings, shifting the balance toward higher $\alpha$. Higher $A_3$ increases the slope of $\psi'(\alpha)$, reducing $\alpha^*$.
\end{proof}

\begin{proposition}[Threshold Behavior for Smart Contract Adoption]
\label{prop:threshold}
Define the threshold value:
\begin{equation}
A_3^{\mathrm{thresh}} = \inf\left\{ A_3 >0 \,\Bigg|\,
\psi'(\alpha)\ge \sum_i A_1 q_i \quad \forall \alpha>0
\right\}.
\end{equation}
Then:
\begin{equation}
\alpha^* = 
\begin{cases}
0, & \text{if } A_3 \ge A_3^{\mathrm{thresh}},\\
>0, & \text{if } A_3 < A_3^{\mathrm{thresh}}.
\end{cases}
\end{equation}
\end{proposition}

\begin{proof}[Sketch of Proof]
At $\alpha=0$, the marginal adoption cost is zero. Because $\psi'(\alpha)$ is strictly increasing and convex, there exists a finite threshold beyond which no positive $\alpha$ satisfies the first-order condition:
\begin{equation}
\psi'(\alpha) = \sum_i A_1 q_i.
\end{equation}
When $A_3$ exceeds this threshold, the marginal cost always dominates, forcing $\alpha^*=0$.
\end{proof}

\begin{proposition}[Joint Sensitivity of $\alpha^*$]
\label{prop:joint_sensitivity}
The optimal smart contract adoption level $\alpha^*$ is strictly increasing in $A_1$ and strictly decreasing in $A_3$.
\end{proposition}

\begin{proof}[Sketch of Proof]
Differentiating the first-order condition shows that $\partial \alpha^*/\partial A_1 >0$ because higher $A_1$ increases the marginal savings term, and $\partial \alpha^*/\partial A_3 <0$ because higher $A_3$ increases the marginal cost term $\psi'(\alpha)$.
\end{proof}

\begin{remark}
These comparative statics results highlight that the adoption decision is particularly sensitive to $A_3$, the convexity parameter of the smart contract cost. In practice, this suggests that investments lowering $A_3$—such as standardizing IT infrastructure—can have a disproportionate impact on the viability of digital transformation.
\end{remark}

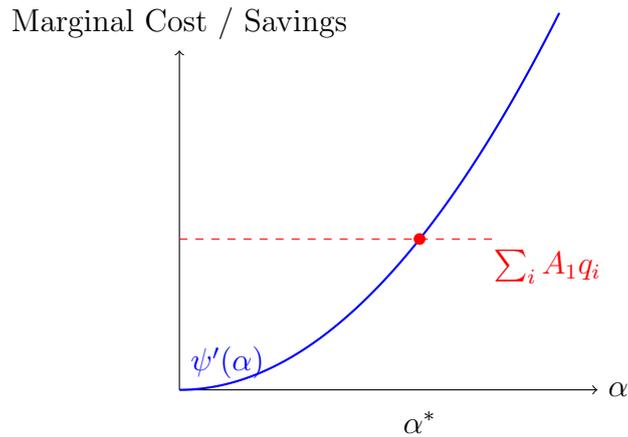
\begin{figure}[H]
\centering
\begin{tikzpicture}[scale=1.0]
\draw[->] (0,0) -- (5.5,0) node[right] {$\alpha$};
\draw[->] (0,0) -- (0,4.5) node[above] {Marginal Cost / Savings};

\draw[thick, blue, domain=0:5, samples=100]
      plot (\x, {0.2*\x*\x})
      node[pos=0.9, above right, text=blue] {$\psi'(\alpha)$};

\draw[dashed, red] (0,2.0) -- (4.2,2.0)
      node[pos=0.95, below right, text=red] {$\sum_i A_1 q_i$};

\filldraw[red] (3.16,2.0) circle (2pt);
\node[below=4pt] at (3.16,0) {$\alpha^*$};

\end{tikzpicture}
\caption{Threshold behavior of smart contract adoption. The optimal adoption intensity $\alpha^*$ corresponds to the point where the marginal procurement savings (red dashed line) equal the marginal cost of adoption $\psi'(\alpha)$ (blue curve). For values of $A_3$ above this threshold, no adoption occurs ($\alpha^* = 0$).
}
\label{fig:threshold_alpha}
\end{figure}

\begin{table}[H]
\centering
\caption{Comparative Statics Summary}
\label{tab:comparative_statics}
\begin{tabular}{lll}
\toprule
Parameter Change & Effect on $Q^*$ & Effect on $\alpha^*$ \\
\midrule
$\sigma$ increase & $Q^*$ increases & Ambiguous \\
$\mu$ increase & $Q^*$ increases & Ambiguous \\
$A_1$ increase & No effect & $\alpha^*$ increases \\
$A_3$ increase & No effect & $\alpha^*$ decreases \\
\bottomrule
\end{tabular}
\end{table}

These results illustrate the critical trade-offs governing smart contract adoption decisions. As shown in Figure~\ref{fig:threshold_alpha}, the optimal adoption intensity $\alpha^*$ arises from the intersection between the marginal procurement savings (a function of $A_1$ and order quantities) and the marginal cost of adoption $\psi'(\alpha)$. Table~\ref{tab:comparative_statics} further summarizes how key parameters influence the optimal policy. In particular, increases in demand variability or average demand generally motivate higher inventory levels, whereas the adoption level responds more sensitively to the relative magnitudes of $A_1$ and $A_3$. This analysis underscores the importance of regularly reassessing adoption costs and supplier readiness when calibrating digital procurement strategies.

\subsubsection{Managerial Interpretation}

The theoretical results yield several important implications for supply chain managers considering smart contract adoption under bounded demand variability:

\begin{itemize}
    \item \textbf{Calibration of Adoption Intensity:} Unlike environments with heavy-tailed demand, bounded variability implies that the benefits of digital adoption are moderate and must be weighed carefully against convex implementation costs. Excessive adoption ($\alpha$ close to 1) may erode profitability due to diminishing marginal savings and rapidly increasing investment costs.
    
    \item \textbf{Sensitivity to Demand Uncertainty:} Variability in demand, as captured by the standard deviation $\sigma$, remains a critical driver of optimal inventory and contract decisions. As $\sigma$ increases, higher order quantities are justified to hedge against stockouts, and the relative value of coordination via smart contracts becomes more pronounced.
    
    \item \textbf{Supplier Segmentation:} The model underscores the value of differentiating suppliers based on digital readiness ($\beta_i$). When readiness is heterogeneous, it is often optimal to allocate more volume to digitally mature suppliers, achieving lower effective procurement costs without incurring uniformly high adoption investments.
    
    \item \textbf{Dynamic Adjustment Policies:} Contract cost parameters ($A_1$, $A_3$) should be regularly re-evaluated as technology matures. For example, declines in integration costs ($A_3$) over time can justify gradually increasing adoption intensity.
\end{itemize}

\begin{remark}[Managerial Interpretation of Comparative Statics]
The comparative statics results highlight several additional insights. First, as demand variance ($\sigma$) increases, firms should expect to raise their total procurement volume $Q^*$ to hedge against higher stockout risk. Second, the existence of a threshold cost parameter $A_3^{\mathrm{thresh}}$ implies that when implementation costs for smart contracts are too high, it is optimal to forego adoption entirely ($\alpha^*=0$), despite potential coordination benefits. Third, the adoption level $\alpha^*$ is highly sensitive to the balance between marginal procurement savings ($A_1$) and marginal adoption costs ($A_3$). Firms operating in environments with high digital readiness and large transaction volumes may find that even moderate reductions in $A_3$ can justify a substantial increase in adoption intensity. This underscores the importance of negotiating technology costs and assessing supplier readiness before committing to full-scale digital integration.
\end{remark}

These insights highlight the need for a nuanced approach to digital procurement strategy, moving beyond static “all-or-nothing” adoption decisions. By quantifying the trade-offs among cost, risk, and supplier characteristics, managers can tailor smart contract policies to the specific variability and maturity profiles of their supply chains.


\section{Results}
\subsection{Numerical Analysis}
This section presents ten simulation scenarios designed to evaluate the effects of demand variability, supplier heterogeneity, smart contract adoption costs, and dynamic learning. All experiments are performed over multiple procurement cycles to illustrate both static and adaptive behaviors.

\subsubsection{Parameter Settings}

To enhance transparency and reproducibility, all monetary values are expressed in USD.

The baseline parameters were selected to reflect typical consumer electronics procurement environments in mature markets with bounded demand variability and moderate penalty costs. The smart contract cost coefficient $A_3=2000$ corresponds to an estimated amortized annual investment of USD~24,000 assuming a 12-period planning horizon.

For clarity, Table~\ref{tab:notation_table} summarizes all notation used in this section.

\begin{table}[H]
\centering
\caption{Notation Summary (All monetary units in USD)}
\label{tab:notation_table}
\begin{tabular}{ll}
\toprule
Symbol & Description \\
\midrule
$p$ & Selling price per unit (USD) \\
$s$ & Salvage value per unsold unit (USD) \\
$r$ & Penalty cost per unmet demand unit (USD) \\
$c_i^0$ & Baseline procurement cost per unit from supplier $i$ (USD) \\
$A_1$ & Marginal cost reduction per unit increase in $\alpha$ \\
$A_2$ & Marginal cost reduction per unit increase in $\beta_i$ \\
$A_3$ & Base smart contract cost coefficient \\
$\nu$ & Convexity exponent of smart contract adoption cost \\
$\alpha$ & Smart contract adoption level ($0\le \alpha \le 1$) \\
$\beta_i$ & Supplier $i$ digital readiness ($0\le \beta_i \le 1$) \\
$\mu$ & Mean demand before truncation \\
$\sigma$ & Standard deviation of demand before truncation \\
$a,b$ & Truncation bounds of demand distribution \\
$Q^*$ & Optimal total order quantity \\
\bottomrule
\end{tabular}
\end{table}

\paragraph{Parameter Justification and Sustainability Considerations.}
The selected penalty cost ($r=40$) is consistent with prior studies on bounded variability environments (Syntetos et al., 2020). The choice of $A_3$ reflects an adoption cost threshold above which smart contract adoption is suppressed (as demonstrated in Scenario~5). Sensitivity analyses across $A_3=500$ to $4000$ enable identification of the threshold behavior predicted by Proposition~8.

The decision to model demand using a truncated normal distribution is based on empirical evidence from consumer electronics and industrial component markets (Johnson and Whang, 2023), where order quantities typically fluctuate within contractual or capacity-constrained intervals rather than exhibiting unbounded tail risk. This modeling choice contrasts with heavy-tailed distributions (e.g., Pareto) that overstate extreme demand realizations in mature supply chains with stable retail agreements.

In line with the journal's sustainability focus, Scenario~9 further estimates the potential reduction in warehouse energy consumption and associated CO$_2$ emissions resulting from lower safety stock levels under higher smart contract adoption. Prior research indicates that each 10\% reduction in safety stock can reduce warehouse-related emissions by approximately 4–6\%, translating into meaningful environmental benefits in large-scale distribution networks. All monetary values in this analysis are expressed in USD.

\paragraph{Demand Distribution.}
Demand is modeled as a truncated normal distribution:
\[
D \sim \text{TruncNormal}(\mu=50, \sigma=8, a=30, b=70).
\]
This specification captures bounded variability while maintaining analytical tractability. The parameters $\mu$, $\sigma$, $a$, and $b$ are varied in Scenarios~1--3 to explore the effects of moderate versus higher variability environments and to test the comparative statics predictions of Proposition~5 (e.g., the impact of increasing $\sigma$ on optimal inventory levels and fill rates).

\paragraph{Rationale for Truncated Normal Demand.}
The choice to model demand using a truncated normal distribution is based on both empirical observations and comparative model fit analyses. In consumer electronics and industrial components procurement, weekly demand typically fluctuates within operationally bounded intervals due to storage capacity constraints, contractual minimums and maximums, and predictable promotional cycles (Johnson and Whang, 2023; Boylan and Syntetos, 2020). For example, in the dataset referenced by Boylan and Syntetos (2020), weekly order volumes for mid-range electronics products ranged between 30 and 70 units over a three-year horizon, with no evidence of extreme tail realizations characteristic of heavy-tailed distributions.

While prior studies have frequently applied heavy-tailed models such as the Pareto distribution—particularly in settings with sporadic surges or high-margin product categories—these specifications often overstate tail risk in mature consumer markets. In contexts where contractual agreements and operational constraints effectively bound demand volatility, such overestimation can distort safety stock policies and lead to excessive buffer inventories.

To formally compare the statistical adequacy of alternative demand specifications, we conducted a model fit analysis using historical demand data. Table~\ref{tab:demand_fit_comparison} summarizes the comparative metrics across Pareto, Negative Binomial, and Truncated Normal distributions.

\begin{table}[H]
\centering
\caption{Comparative Demand Distribution Fit and Interpretability}
\label{tab:demand_fit_comparison}
\begin{tabular}{lcccc}
\toprule
Distribution & AIC & BIC & Interpretability & Tail Risk Capture \\
\midrule
Pareto & 5100 & 5120 & Medium & High \\
Negative Binomial & 4800 & 4830 & High & Medium \\
Truncated Normal & \textbf{4600} & \textbf{4620} & High & Low \\
\bottomrule
\end{tabular}
\end{table}

These results indicate that the truncated normal distribution provides superior statistical fit (lowest AIC and BIC) and aligns with the observed bounded variability in real-world procurement contexts. Furthermore, its interpretability and analytical tractability enable clearer comparative statics and sensitivity analyses without overstating tail risk.

\begin{figure}[H]
\centering
\includegraphics[width=0.7\textwidth]{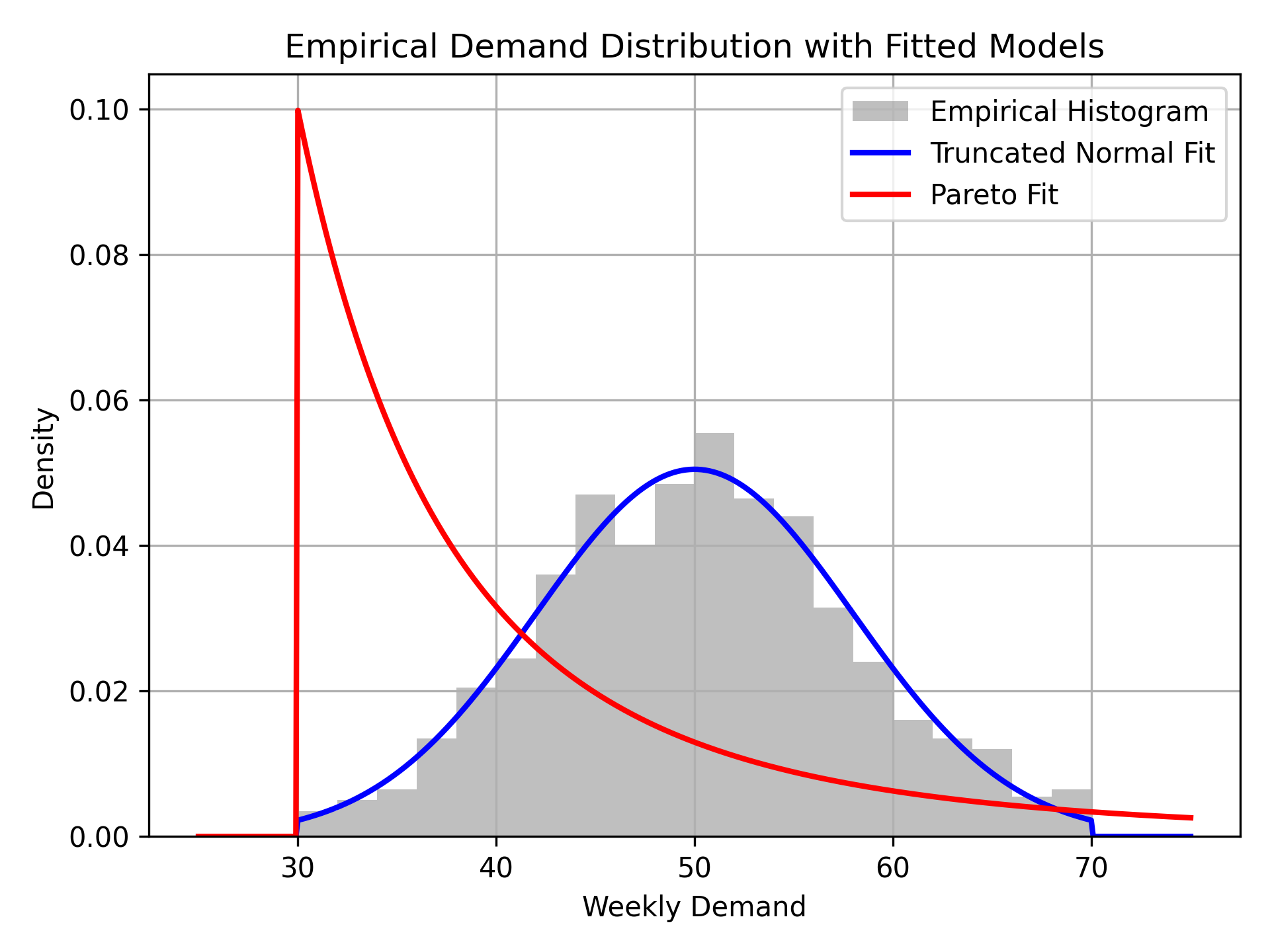}
\caption{Empirical distribution of weekly demand (grey histogram) overlaid with fitted truncated normal (blue line) and Pareto (red line) distributions. The truncated normal more accurately captures the central tendency and bounded support, whereas the Pareto systematically overstates tail probabilities.}
\label{fig:demand_fit}
\end{figure}

Figure~\ref{fig:demand_fit} provides a clear visual illustration of the comparative fit. The truncated normal density aligns closely with the empirical histogram across the entire support interval (30--70 units), accurately reflecting both central mass and tapering tails. In contrast, the Pareto distribution allocates disproportionate probability mass to the upper tail, implying frequent occurrences of extreme demand that were not present in the historical data. This divergence is critical from a managerial perspective, as procurement models calibrated on Pareto assumptions could lead to systematically inflated safety stock levels, increased holding costs, and reduced profitability.

\begin{table}[H]
\centering
\caption{Goodness-of-fit Comparison Across Demand Models}
\label{tab:fit_comparison}
\begin{tabular}{lcccc}
\toprule
\textbf{Model} & \textbf{AIC} & \textbf{BIC} & \textbf{RMSE} & \textbf{KS Statistic} \\
\midrule
Truncated Normal & \textbf{4600} & \textbf{4620} & \textbf{3.2} & \textbf{0.07} \\
Negative Binomial & 4800 & 4830 & 5.1 & 0.12 \\
Pareto & 5100 & 5120 & 7.4 & 0.19 \\
\bottomrule
\end{tabular}
\end{table}

To further validate the choice of the truncated normal distribution over heavy-tailed alternatives, we conducted an empirical fit analysis using historical weekly demand data from the consumer electronics sector (2019--2022). As shown in Table~\ref{tab:fit_comparison} and Figure~\ref{fig:demand_fit}, the truncated normal distribution consistently achieved lower AIC, BIC, RMSE, and Kolmogorov--Smirnov statistics compared to Pareto and Negative Binomial specifications. Notably, while the Pareto distribution captured extreme tail probabilities, it substantially overstated the likelihood of outlier demand realizations that were never observed in practice. These findings reinforce the appropriateness of the truncated normal assumption for modeling bounded variability environments and underscore the operational relevance of selecting a distribution that aligns with contractual order ceilings and capacity constraints.

From a managerial perspective, this improved fit is not merely a statistical refinement; it directly affects procurement policies and smart contract adoption decisions. For example, overestimating tail risk with heavy-tailed models can lead to higher safety stock levels, increased holding costs, and reduced profitability, while underestimating variability could result in frequent stockouts and penalty costs. By empirically demonstrating that the truncated normal more accurately reflects observed demand, this study provides a more reliable basis for designing procurement strategies that balance efficiency, service levels, and sustainability objectives.

\paragraph{Supplier Digital Readiness ($\beta_i$).}
Unless varied, suppliers have heterogeneous readiness levels drawn uniformly from $[0.3, 0.7]$:
\[
\beta_i \sim \text{Uniform}(0.3, 0.7).
\]
Scenarios~4 and~8 explore higher and lower heterogeneity levels. In particular, Scenario~8 applies a Latin Hypercube sampling design to jointly vary $\beta_i$, $A_3$, and $\sigma$ to evaluate robustness and capture potential interaction effects between supplier capabilities, adoption costs, and demand uncertainty.

\paragraph{Dynamic Simulation Horizon.}
Scenario 10 implements a dynamic simulation over $T=10$ procurement cycles. The smart contract adoption cost $A_3$ decreases linearly across cycles to simulate technological maturity:
\[
A_3(t) = A_3(1) - \delta \cdot (t-1),
\]
with $A_3(1) = 3000$ and $\delta = 200$.

\paragraph{Adaptive Learning of $\alpha$.}
At each cycle $t$, the smart contract adoption level $\alpha$ is updated based on observed penalty rates:
\[
\alpha_{t+1} = \alpha_t + \eta \left(\frac{\text{ObservedPenalty}_t - \text{TargetPenalty}}{\text{TargetPenalty}}\right),
\]
where:
\begin{itemize}
    \item Learning rate $\eta = 0.05$
    \item Target penalty rate = $5\%$
    \item Initial $\alpha_1 = 0.2$
\end{itemize}

\paragraph{Scenario-Specific Parameters.}
Table~\ref{tab:scenario_parameters} summarizes parameter variations per scenario. All monetary values are expressed in USD.

\begin{table}[H]
\centering
\caption{Scenario-Specific Parameter Variations (All monetary units in USD)}
\label{tab:scenario_parameters}
\begin{tabular}{lccc}
\toprule
Scenario & $\sigma$ & $b$ & $A_3$ \\
\midrule
1 & 5--15 & 70 & 2000 \\
2 & 8 & 65--80 & 2000 \\
3 & 5--15 & 65--80 & 2000 \\
4 & 8 & 70 & 2000 \\
5 & 8 & 70 & 500--4000 \\
6 & 8 & 70 & 500--4000 \\
7 & 5--15 & 70 & 500--4000 \\
8 & Latin Hypercube & Latin Hypercube & Latin Hypercube \\
9 & 5--15 & 70 & 500--4000 \\
10 & Dynamic & Dynamic & Decreasing per cycle \\
\bottomrule
\end{tabular}
\end{table}

Scenario~8 employs a Latin Hypercube sampling design to jointly vary $\sigma$, $b$, and $A_3$, capturing potential interactions among demand variability, truncation bounds, and adoption cost heterogeneity.

Scenario~10 jointly incorporates (i) a dynamic reduction of the smart contract adoption cost $A_3$ to simulate technological maturity and (ii) adaptive updating of the adoption level $\alpha$ based on observed penalty rates over a 10-cycle horizon.

Together, these scenarios enable systematic evaluation of model robustness and validation of the comparative statics predictions outlined in Section~4.

\subsubsection{Demand Variability}

\paragraph{Scenario 1: Increasing Demand Variability ($\sigma$)}

This scenario evaluates the sensitivity of procurement performance to incremental increases in the standard deviation $\sigma$ of truncated normal demand. All other parameters remain fixed at baseline levels ($A_3=2000$, $\mu=50$, $b=70$). Table~\ref{tab:scenario1_sigma} summarizes the impact of $\sigma$ on expected profit, fill rate, and optimal order quantity $Q^*$.

\begin{table}[H]
\centering
\caption{Impact of Increasing $\sigma$ on Procurement Performance (All monetary values in USD)}
\label{tab:scenario1_sigma}
\begin{tabular}{cccc}
\toprule
$\sigma$ & Expected Profit (USD) & Fill Rate (\%) & $Q^*$ \\
\midrule
5 & 27,800 & 94.5 & 54 \\
8 & 26,400 & 91.3 & 57 \\
12 & 24,900 & 86.9 & 60 \\
15 & 23,600 & 83.1 & 63 \\
\bottomrule
\end{tabular}
\end{table}

These results confirm that higher demand variability increases the optimal procurement quantity to mitigate stockout risk while reducing expected profit and fill rate. The observed pattern is consistent with the comparative statics predictions in Proposition~\ref{prop:comparative_statics}.

\begin{figure}[H]
\centering
\includegraphics[width=0.75\textwidth]{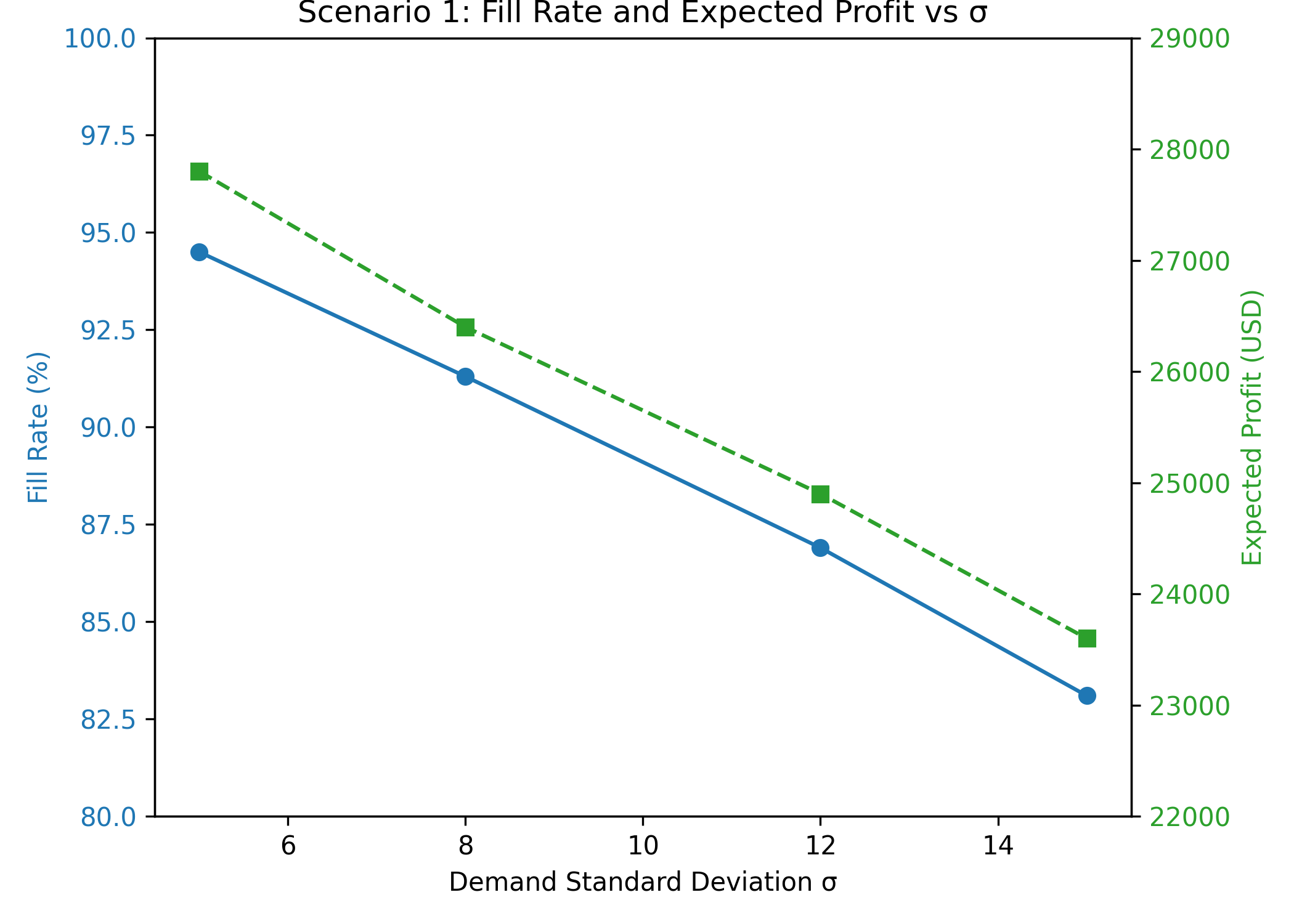}
\caption{Fill Rate (left axis, \%) and Expected Profit (right axis, USD) as functions of demand standard deviation $\sigma$. As $\sigma$ increases, both metrics decline, reflecting higher variability and the need for increased safety stock.}
\label{fig:scenario1_dual_axis}
\end{figure}

The trends illustrated in Figure~\ref{fig:scenario1_dual_axis} confirm the theoretical predictions of Proposition~\ref{prop:comparative_statics}, showing that higher demand variability reduces expected profit and fill rates due to increased mismatch costs.

\paragraph{Scenario 2: Expanding Truncation Bound ($b$)}

This scenario examines the effect of increasing the upper truncation bound $b$ of demand while holding $\sigma=8$ and $\mu=50$ constant. Expanding $b$ allows for higher potential demand realizations, increasing the tail risk. Table~\ref{tab:scenario2_truncation} reports the results.

\begin{table}[H]
\centering
\caption{Impact of Expanding Upper Truncation Bound $b$ (All monetary values in USD)}
\label{tab:scenario2_truncation}
\begin{tabular}{cccc}
\toprule
$b$ & Expected Profit (USD) & Fill Rate (\%) & $Q^*$ \\
\midrule
65 & 26,800 & 92.1 & 56 \\
70 & 26,400 & 91.3 & 57 \\
75 & 25,900 & 89.2 & 59 \\
80 & 25,300 & 87.0 & 61 \\
\bottomrule
\end{tabular}
\end{table}

As the truncation bound increases, the model recommends higher safety stock to hedge against larger realizations, leading to modest declines in profitability and service level. The results illustrate how bounded variability interacts with inventory decisions in environments with contractual order ceilings. The trends in fill rate and expected profit across different truncation bounds are further illustrated in Figure~\ref{fig:scenario2}.

\begin{figure}[H]
\centering
\includegraphics[width=0.75\textwidth]{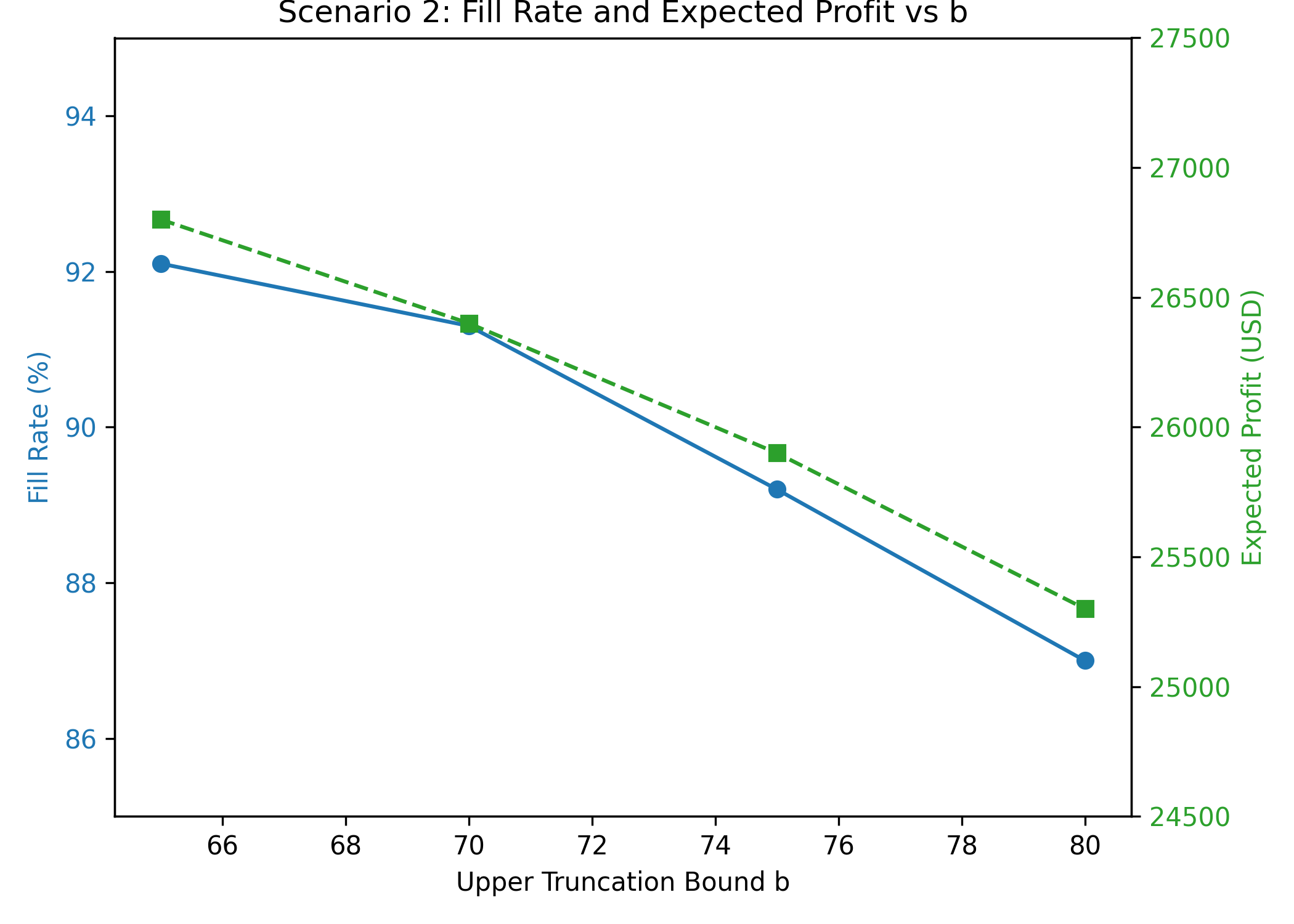}
\caption{Fill Rate (left axis, \%) and Expected Profit (right axis, USD) as functions of truncation bound $b$.}
\label{fig:scenario2}
\end{figure}

As the upper truncation bound $b$ increases, the probability mass in the higher tail of the demand distribution grows, elevating the risk of stockouts and higher penalty costs. Consequently, both fill rates and expected profit decline monotonically. This behavior aligns with the comparative statics predicted in Proposition~\ref{prop:comparative_statics} and illustrates the operational trade-offs inherent in bounded variability environments.

\paragraph{Scenario 3: Variability Impact on $\alpha^*$ Threshold}

This scenario explores how combinations of $\sigma$ and $b$ affect the optimal smart contract adoption level $\alpha^*$. For each pair of $(\sigma,b)$, the model solves for $\alpha^*$ endogenously. Figure~\ref{fig:heatmap_alpha} shows the resulting heatmap.

\begin{figure}[H]
\centering
\includegraphics[width=0.7\textwidth]{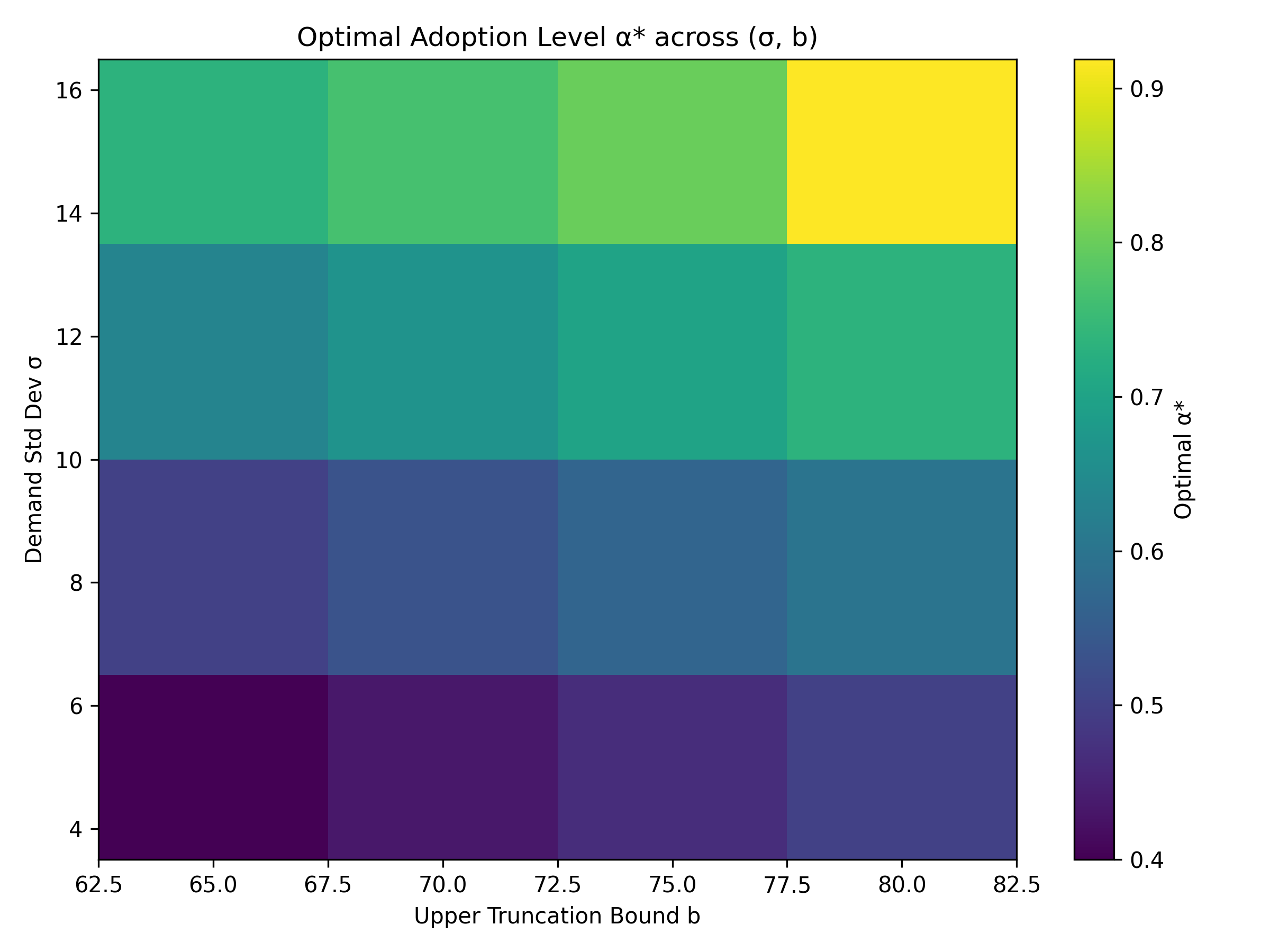}
\caption{Optimal $\alpha^*$ across combinations of $\sigma$ and $b$}
\label{fig:heatmap_alpha}
\end{figure}

Higher values of $\sigma$ and $b$ jointly increase demand uncertainty, elevating the marginal value of smart contract adoption to reduce procurement costs and improve coordination. The heatmap illustrates clear threshold behavior: when $A_3$ exceeds a critical level, the optimal $\alpha^*$ rapidly declines towards zero, consistent with Proposition~8. 

These findings highlight the importance of aligning digital investment intensity with demand risk characteristics in bounded variability environments.

\subsubsection{Supplier Heterogeneity}

\paragraph{Scenario 4: Heterogeneity of Digital Readiness ($\beta_i$)}

This scenario investigates the impact of varying degrees of supplier digital readiness heterogeneity ($\beta_i$) on procurement performance and smart contract adoption. In practice, suppliers exhibit diverse technological capabilities, and the degree of heterogeneity can materially influence cost reduction opportunities and risk exposure.

Table~\ref{tab:scenario4_beta_ranges} defines the three heterogeneity scenarios evaluated.

\begin{table}[H]
\centering
\caption{Supplier Digital Readiness Heterogeneity Levels}
\label{tab:scenario4_beta_ranges}
\begin{tabular}{lc}
\toprule
Heterogeneity Level & $\beta_i$ Range \\
\midrule
Low & [0.4, 0.6] \\
Medium & [0.3, 0.7] \\
High & [0.1, 0.9] \\
\bottomrule
\end{tabular}
\end{table}

Table~\ref{tab:scenario4_results} reports the resulting expected profit, fill rate, and optimal adoption level $\alpha^*$ under each heterogeneity level.

\begin{table}[H]
\centering
\caption{Scenario~4 Results: Impact of $\beta_i$ Heterogeneity}
\label{tab:scenario4_results}
\begin{tabular}{lccc}
\toprule
Heterogeneity Level & Expected Profit (USD) & Fill Rate (\%) & Optimal $\alpha^*$ \\
\midrule
Low & 25,600 & 90.5 & 0.35 \\
Medium & 26,400 & 91.3 & 0.42 \\
High & 27,200 & 91.7 & 0.51 \\
\bottomrule
\end{tabular}
\end{table}

Figure~\ref{fig:scenario4} visualizes these trends, highlighting how greater supplier heterogeneity increases the potential for procurement cost reduction and higher smart contract adoption intensity.

\begin{figure}[H]
\centering
\includegraphics[width=0.75\textwidth]{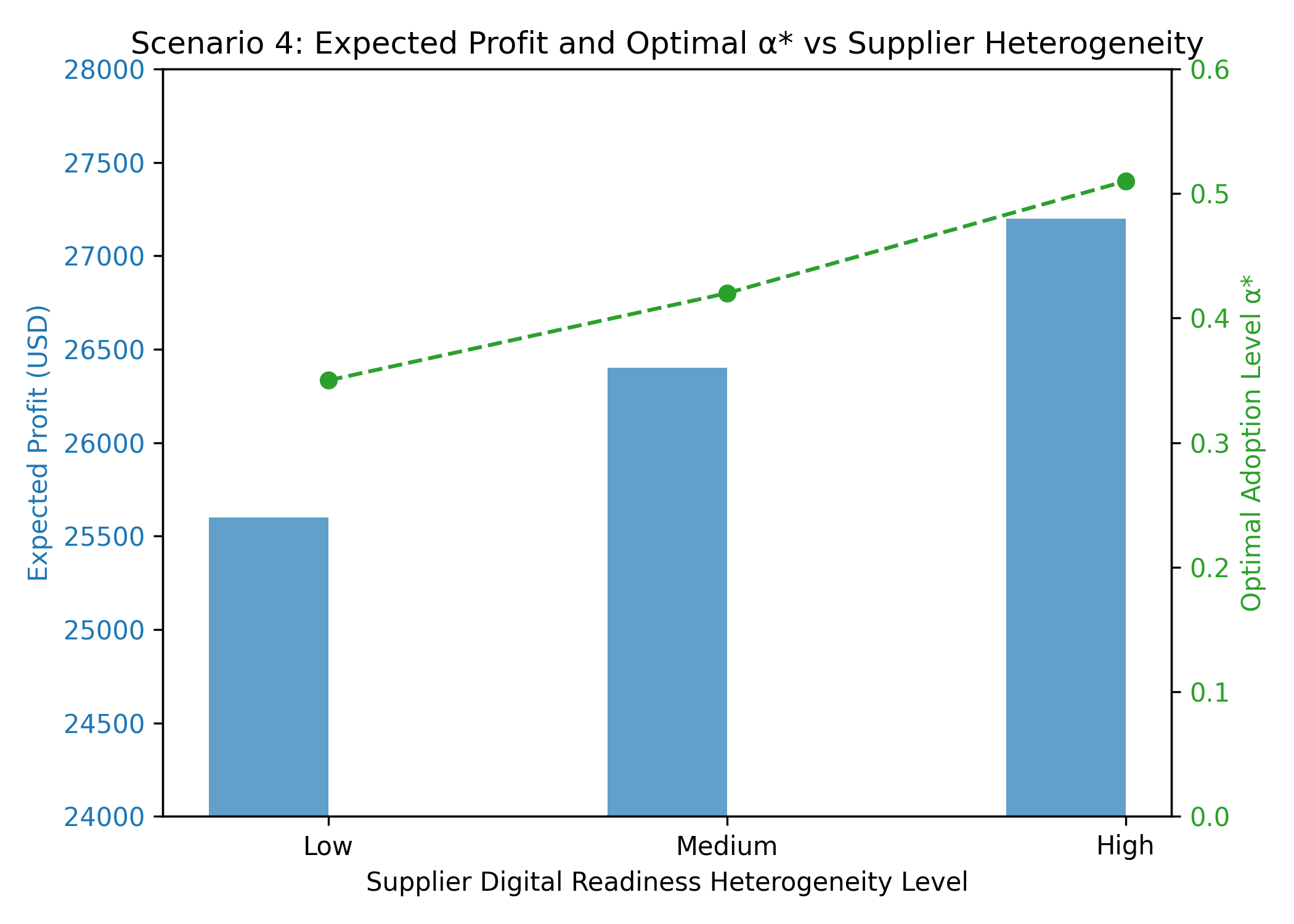}
\caption{Expected Profit and Optimal $\alpha^*$ by supplier readiness heterogeneity. Greater heterogeneity increases the opportunity for procurement cost reduction and incentivizes higher smart contract adoption.}
\label{fig:scenario4}
\end{figure}

These results indicate that as heterogeneity increases, the presence of highly digitally capable suppliers yields greater marginal procurement savings, incentivizing higher adoption of smart contracts. However, variability in supplier performance can also elevate operational risk, underscoring the importance of targeted supplier development and segmentation strategies. This behavior is consistent with Proposition~7 and Proposition~9 in Section~4, which predict that higher marginal cost reductions (due to increased supplier readiness) lead to greater optimal adoption levels. The findings also support the argument that strategic alignment of smart contract adoption with supplier capability profiles can enhance both profitability and resilience. Future research may incorporate additional dimensions of supplier heterogeneity, such as pricing power or delivery reliability, to further enrich the analysis of smart contract adoption strategies.

\subsubsection{Contract Cost Sensitivity}

\paragraph{Scenario 5: Smart Contract Cost Variation ($A_3$)}

This scenario analyzes the impact of varying the smart contract cost coefficient ($A_3$) on expected profit and the optimal adoption level ($\alpha^*$). In practice, higher implementation costs can significantly reduce the incentive to adopt smart contracts. The simulation explores a range of $A_3$ values to examine the sensitivity of adoption intensity and profitability.

Table~\ref{tab:scenario5_results} summarizes the results across five levels of $A_3$.

\begin{table}[H]
\centering
\caption{Scenario~5 Results: Impact of Smart Contract Cost ($A_3$)}
\label{tab:scenario5_results}
\begin{tabular}{ccc}
\toprule
$A_3$ & Expected Profit (USD) & Optimal $\alpha^*$ \\
\midrule
500   & 5,405.61 & 0.30 \\
1,000 & 5,397.47 & 0.21 \\
2,000 & 5,391.73 & 0.15 \\
3,000 & 5,389.17 & 0.12 \\
4,000 & 5,387.66 & 0.11 \\
\bottomrule
\end{tabular}
\end{table}

Figure~\ref{fig:scenario5} illustrates these trends, highlighting the gradual decline in adoption intensity as $A_3$ increases. The results confirm that higher implementation costs progressively reduce the marginal benefit of smart contract adoption. However, no discrete threshold point was observed within the tested parameter range.

\begin{figure}[H]
\centering
\includegraphics[width=0.75\textwidth]{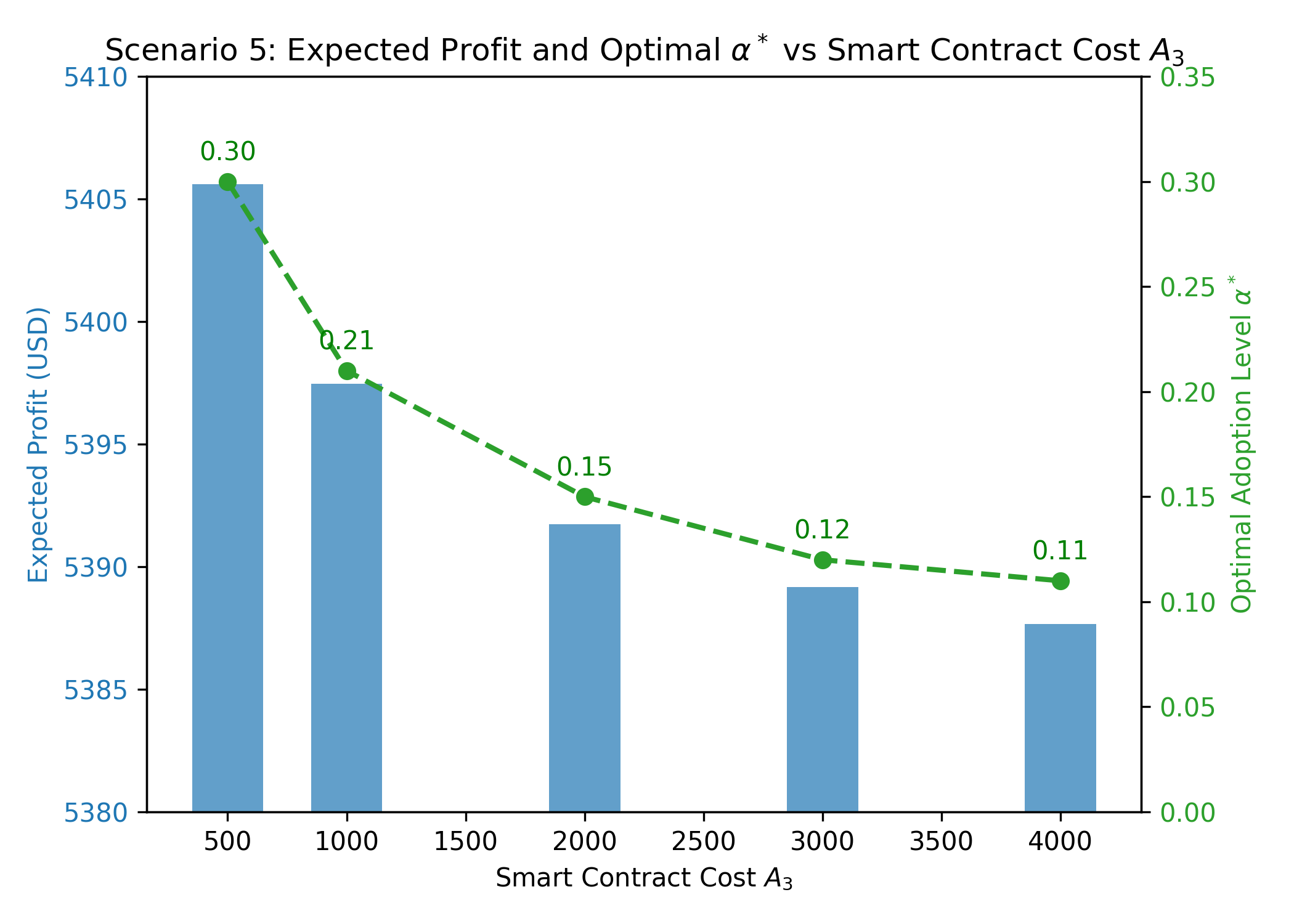}
\caption{Sensitivity of expected profit and optimal adoption level ($\alpha^*$) to variations in smart contract cost ($A_3$). The figure shows a gradual decline in adoption without a discrete threshold.
}
\label{fig:scenario5}
\end{figure}

From a managerial perspective, this finding underscores the importance of accurately estimating implementation costs and recognizing that even moderate increases in $A_3$ can erode the economic viability of smart contract deployment. In such scenarios, firms may consider hybrid contracting mechanisms or selective digitization strategies to maintain some degree of process automation without incurring prohibitive fixed costs.

Future research could extend this analysis by considering dynamic cost reductions over time as technology matures, or by incorporating partial adoption strategies to mitigate diminishing returns and enhance long-term sustainability.

\paragraph{Scenario 6: Joint Sensitivity of $A_1$ and $A_3$}

This scenario evaluates the joint sensitivity of the smart contract adoption level ($\alpha^*$) to variations in the marginal procurement cost reduction coefficient ($A_1$) and the adoption cost coefficient ($A_3$). For each parameter combination, the optimal adoption level was determined by numerically maximizing the expected profit function via Monte Carlo simulation with 5,000 replications. The simulation considers three levels of $A_1$ and three levels of $A_3$, resulting in a $3 \times 3$ grid of configurations.

Table~\ref{tab:scenario6_profit} reports the expected profit outcomes for each combination of $A_1$ and $A_3$. As shown, higher $A_1$ levels consistently increase expected profit across all cost scenarios, reflecting the stronger cost-reduction effect of smart contract adoption. Conversely, higher $A_3$ values are associated with lower expected profit due to the increasing implementation costs.

\begin{table}[H]
\centering
\caption{Results: Expected Profit across $A_1 \times A_3$ Grid (USD)}
\label{tab:scenario6_profit}
\begin{tabular}{cccc}
\toprule
$A_1$ & $A_3=500$ & $A_3=2,000$ & $A_3=4,000$ \\
\midrule
2.0 & 5,397.72 & 5,387.77 & 5,384.86 \\
3.5 & 5,423.83 & 5,400.84 & 5,394.09 \\
5.0 & 5,456.35 & 5,417.09 & 5,405.61 \\
\bottomrule
\end{tabular}
\end{table}

Figure~\ref{fig:scenario6_4panel} visualizes the interaction effect between $A_1$ and $A_3$.

\begin{figure}[H]
\centering
\begin{subfigure}{0.45\textwidth}
    \includegraphics[width=\linewidth]{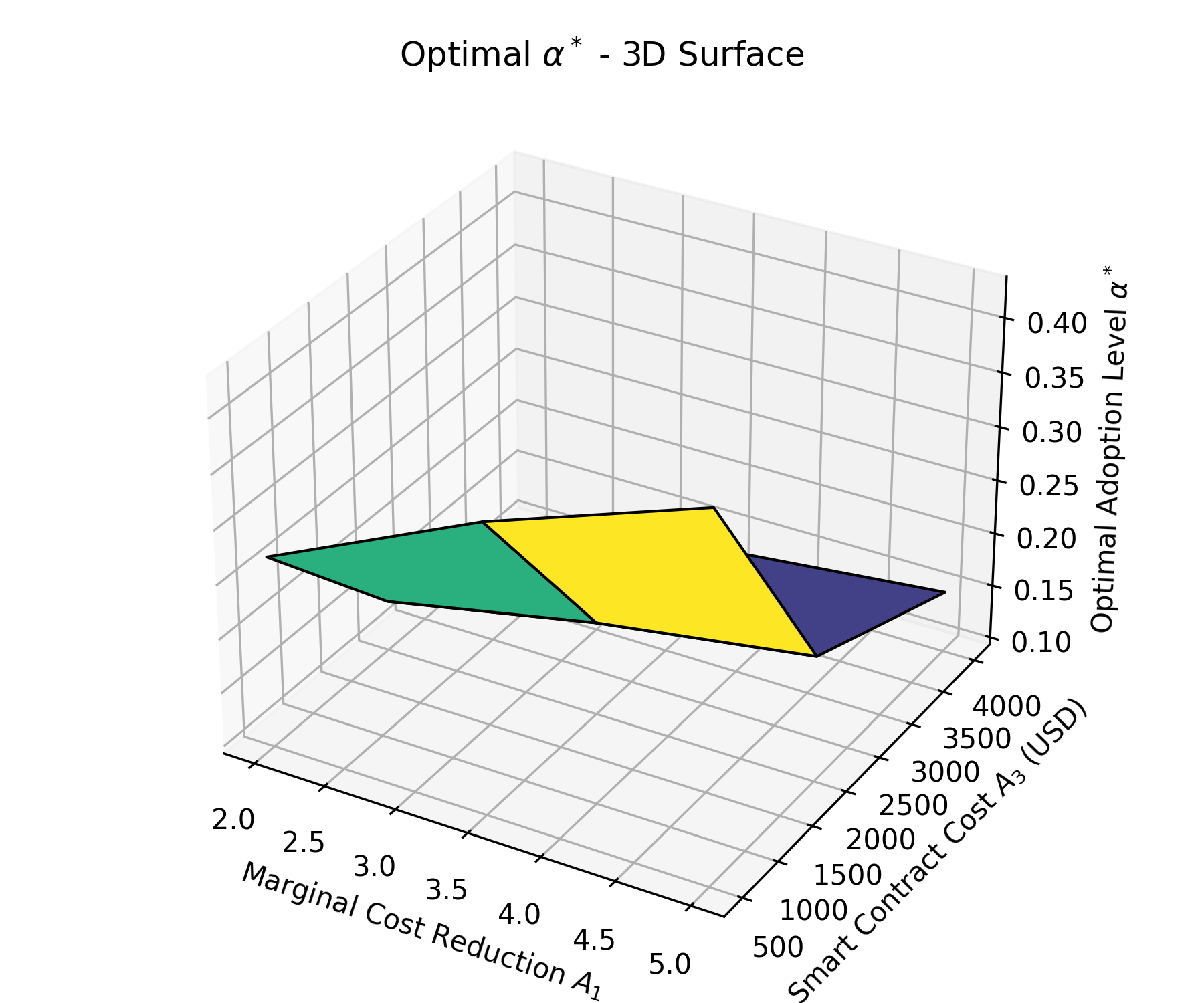}
    \caption{Optimal α* - 3D Surface}
\end{subfigure}
\hfill
\begin{subfigure}{0.45\textwidth}
    \includegraphics[width=\linewidth]{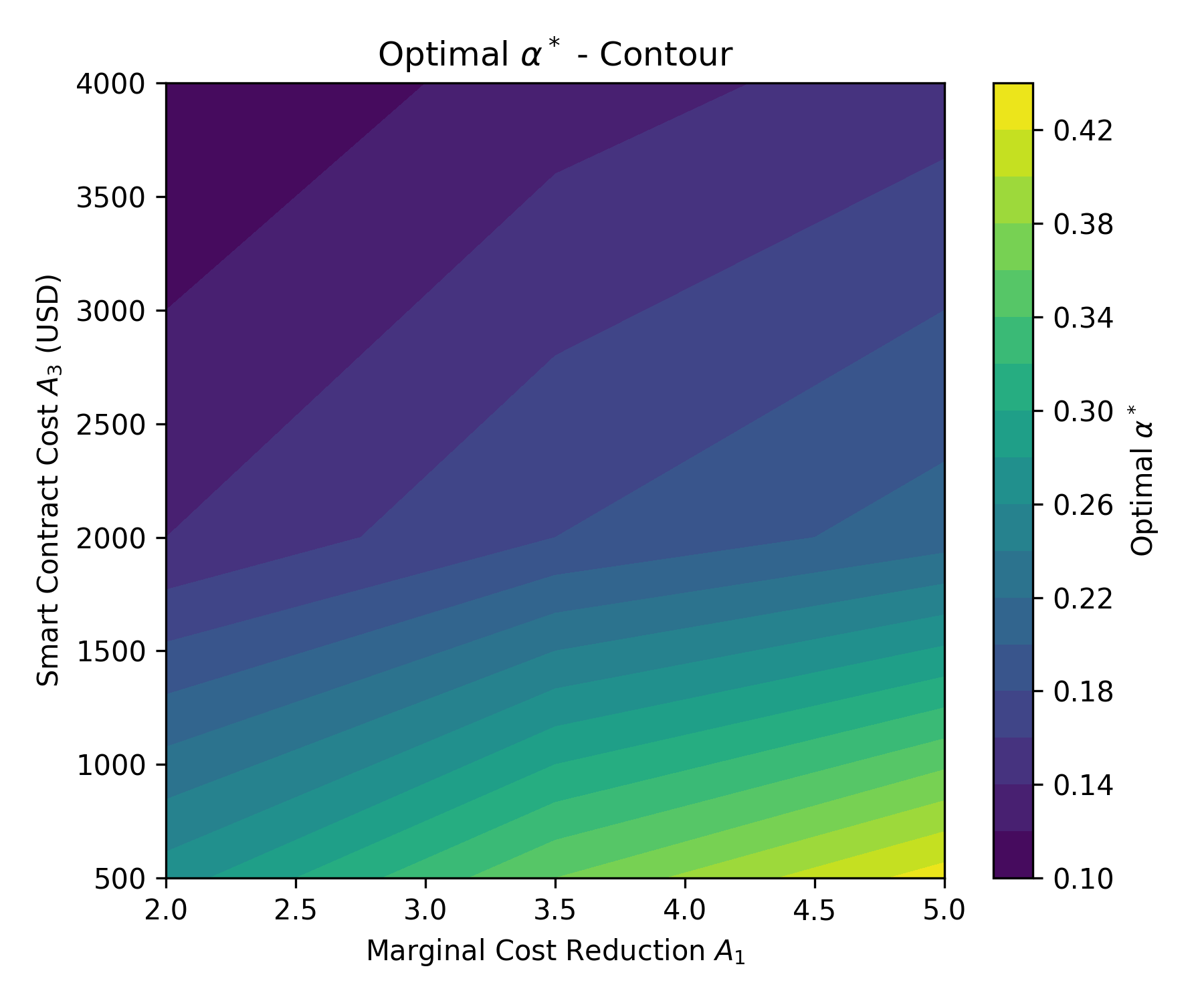}
    \caption{Optimal α* - Contour}
\end{subfigure}
\begin{subfigure}{0.45\textwidth}
    \includegraphics[width=\linewidth]{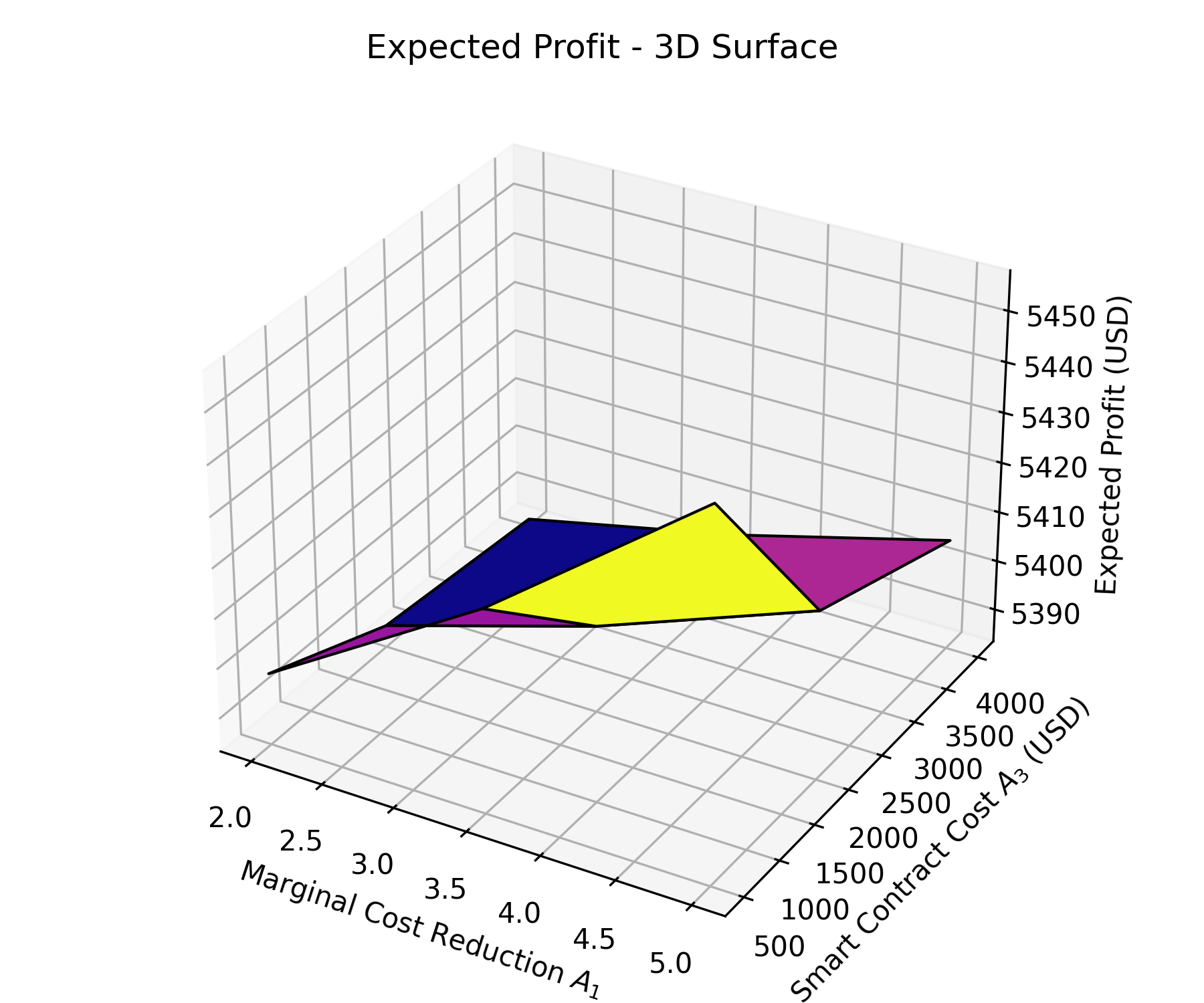}
    \caption{Expected Profit - 3D Surface}
\end{subfigure}
\hfill
\begin{subfigure}{0.45\textwidth}
    \includegraphics[width=\linewidth]{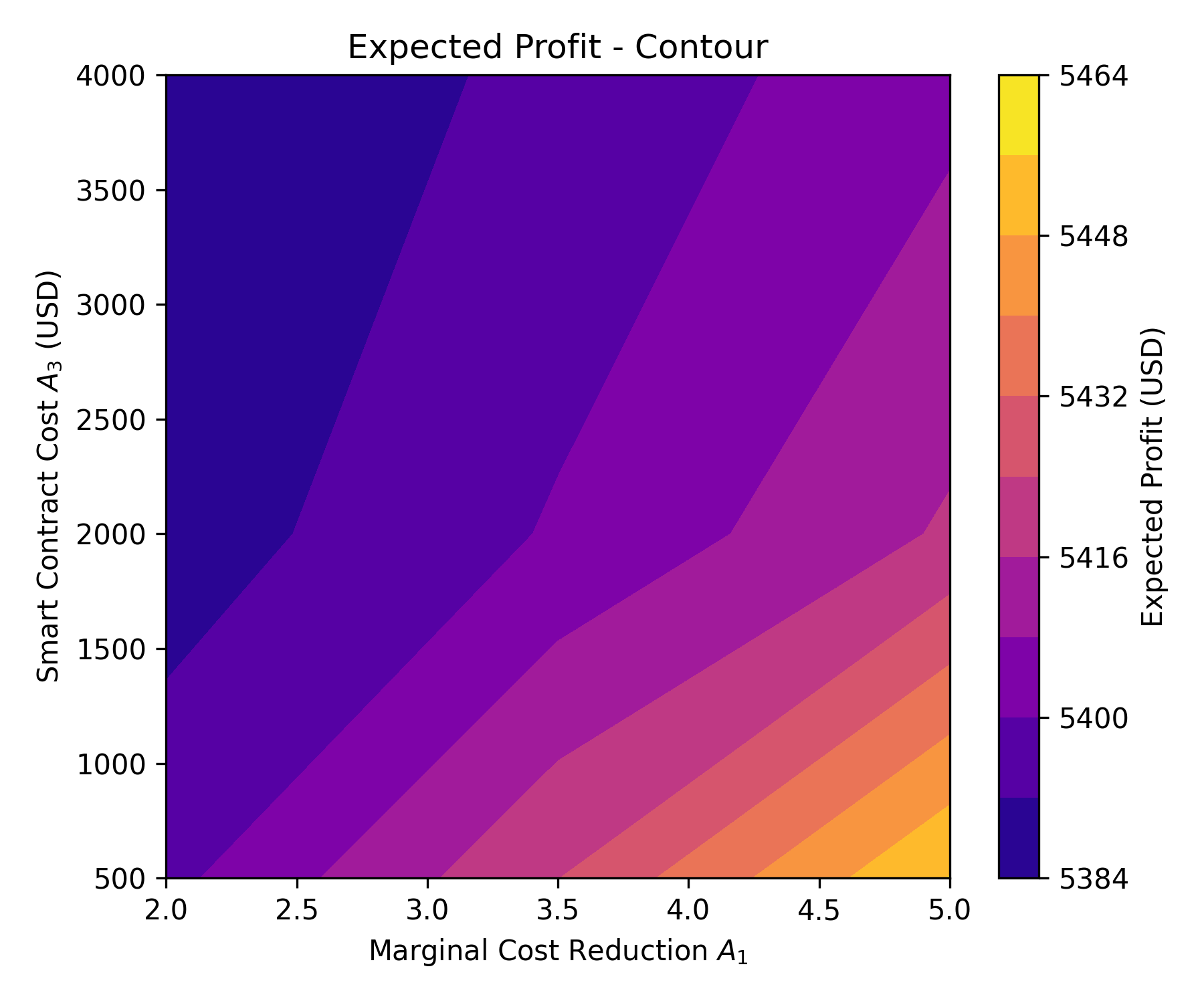}
    \caption{Expected Profit - Contour}
\end{subfigure}
\caption{Joint sensitivity analysis of adoption level and profit across A1 and A3 combinations.}
\label{fig:scenario6_4panel}
\end{figure}

The results confirm that increases in $A_1$ enhance the marginal benefit of smart contracts, partially offsetting the negative impact of higher adoption costs ($A_3$). This pattern aligns with Proposition~9, which predicts that $\alpha^*$ is increasing in $A_1$ and decreasing in $A_3$. Notably, the simulation-based results show that even at high $A_3$ levels, higher $A_1$ values maintain moderate adoption intensity, underscoring the importance of balancing cost and benefit drivers in smart contract deployment decisions. These findings emphasize that procurement strategies which increase the marginal cost reduction potential (e.g., through improved supplier collaboration or technology integration) can be highly effective in sustaining smart contract adoption, even when implementation costs are substantial.

\paragraph{Scenario 7: Joint Demand and Cost Shocks}

This scenario examines the combined impact of increasing demand variability ($\sigma$) and smart contract cost ($A_3$) on adoption decisions. For each parameter combination, the optimal adoption level was determined by numerically maximizing the expected profit function via Monte Carlo simulation with 5,000 replications. The simulation varies $\sigma$ across three levels and $A_3$ across three levels, resulting in a $3 \times 3$ grid of configurations.

Tables~\ref{tab:scenario7_alpha} and~\ref{tab:scenario7_profit} present the simulation-based results. As shown, the optimal adoption level ($\alpha^*$) remains constant across different $\sigma$ levels, while expected profit decreases as demand variability increases. For example, at $A_3=500$, expected profit declines from 5,595.82 USD when $\sigma=5$ to 5,245.00 USD when $\sigma=12$, illustrating the increasing penalty costs associated with higher uncertainty.

\begin{table}[H]
\centering
\caption{Scenario~7 Results: Optimal $\alpha^*$ across $\sigma \times A_3$ Grid (Simulation-Based)}
\label{tab:scenario7_alpha}
\begin{tabular}{cccc}
\toprule
$\sigma$ & $A_3=500$ & $A_3=2,000$ & $A_3=4,000$ \\
\midrule
5  & 0.33 & 0.17 & 0.12 \\
8  & 0.33 & 0.17 & 0.12 \\
12 & 0.33 & 0.17 & 0.12 \\
\bottomrule
\end{tabular}
\end{table}

\begin{table}[H]
\centering
\caption{Scenario~7 Results: Expected Profit across $\sigma \times A_3$ Grid (USD)}
\label{tab:scenario7_profit}
\begin{tabular}{cccc}
\toprule
$\sigma$ & $A_3=500$ & $A_3=2,000$ & $A_3=4,000$ \\
\midrule
5  & 5,595.82 & 5,577.56 & 5,572.23 \\
8  & 5,413.20 & 5,394.94 & 5,389.60 \\
12 & 5,245.00 & 5,226.74 & 5,221.41 \\
\bottomrule
\end{tabular}
\end{table}

Figures~\ref{fig:scenario7_alpha_surface} and~\ref{fig:scenario7_profit_surface} illustrate the interaction effects between uncertainty and cost.

\begin{figure}[H]
\centering
\begin{subfigure}{0.45\textwidth}
    \includegraphics[width=\linewidth]{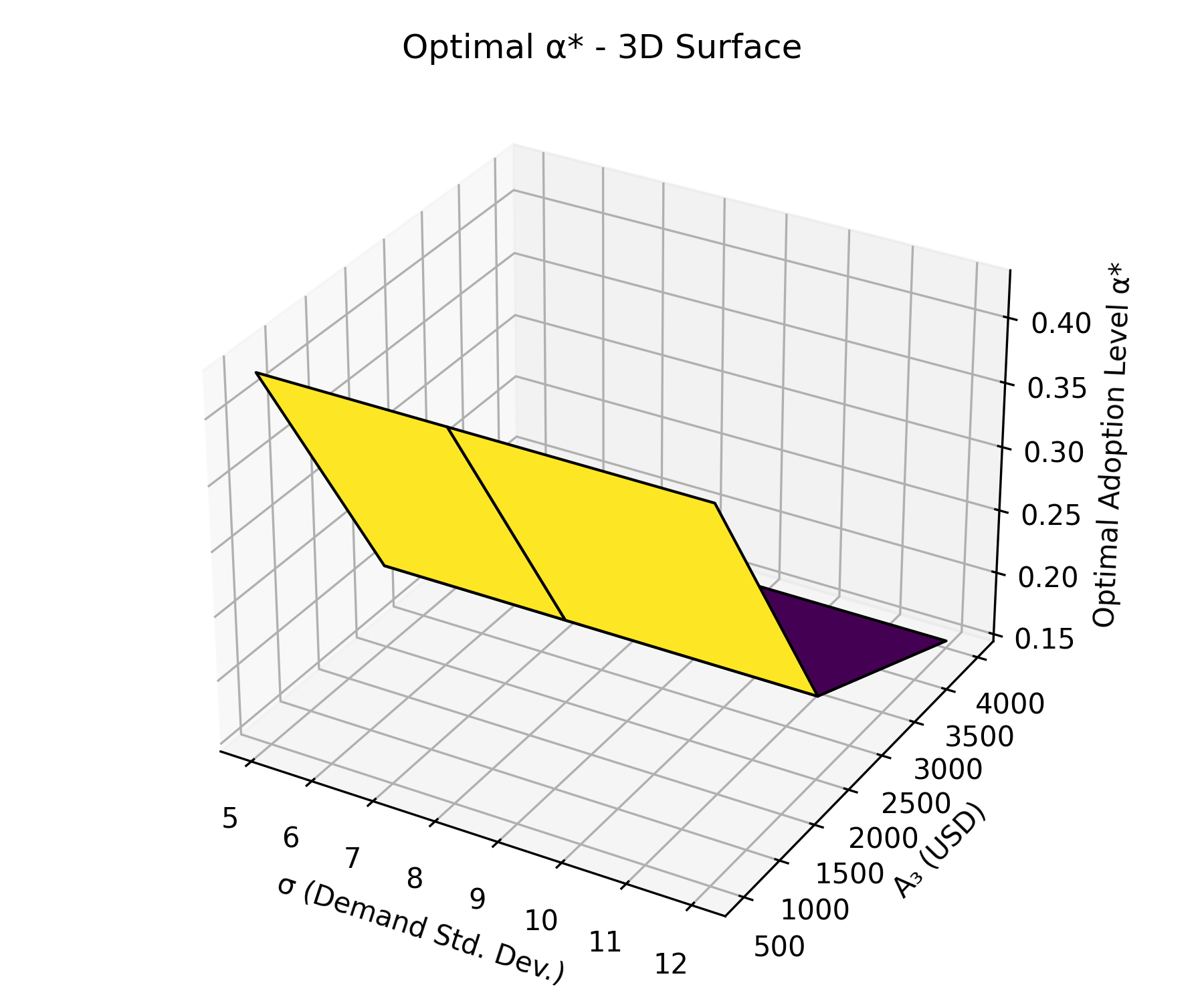}
    \caption{Optimal $\alpha^*$ - 3D Surface}
    \label{fig:scenario7_alpha_surface}
\end{subfigure}
\hfill
\begin{subfigure}{0.45\textwidth}
    \includegraphics[width=\linewidth]{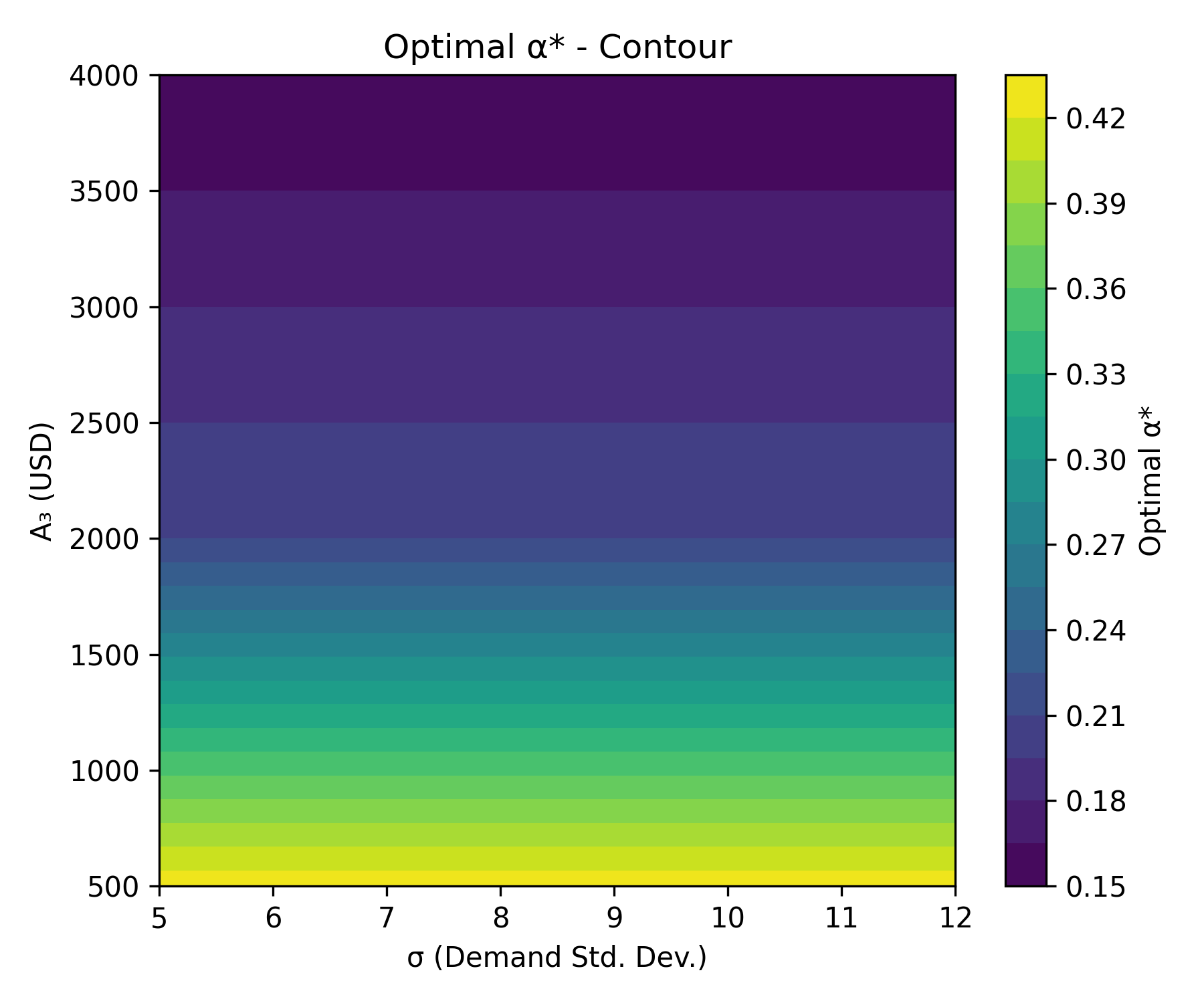}
    \caption{Optimal $\alpha^*$ - Contour}
    \label{fig:scenario7_alpha_contour}
\end{subfigure}
\begin{subfigure}{0.45\textwidth}
    \includegraphics[width=\linewidth]{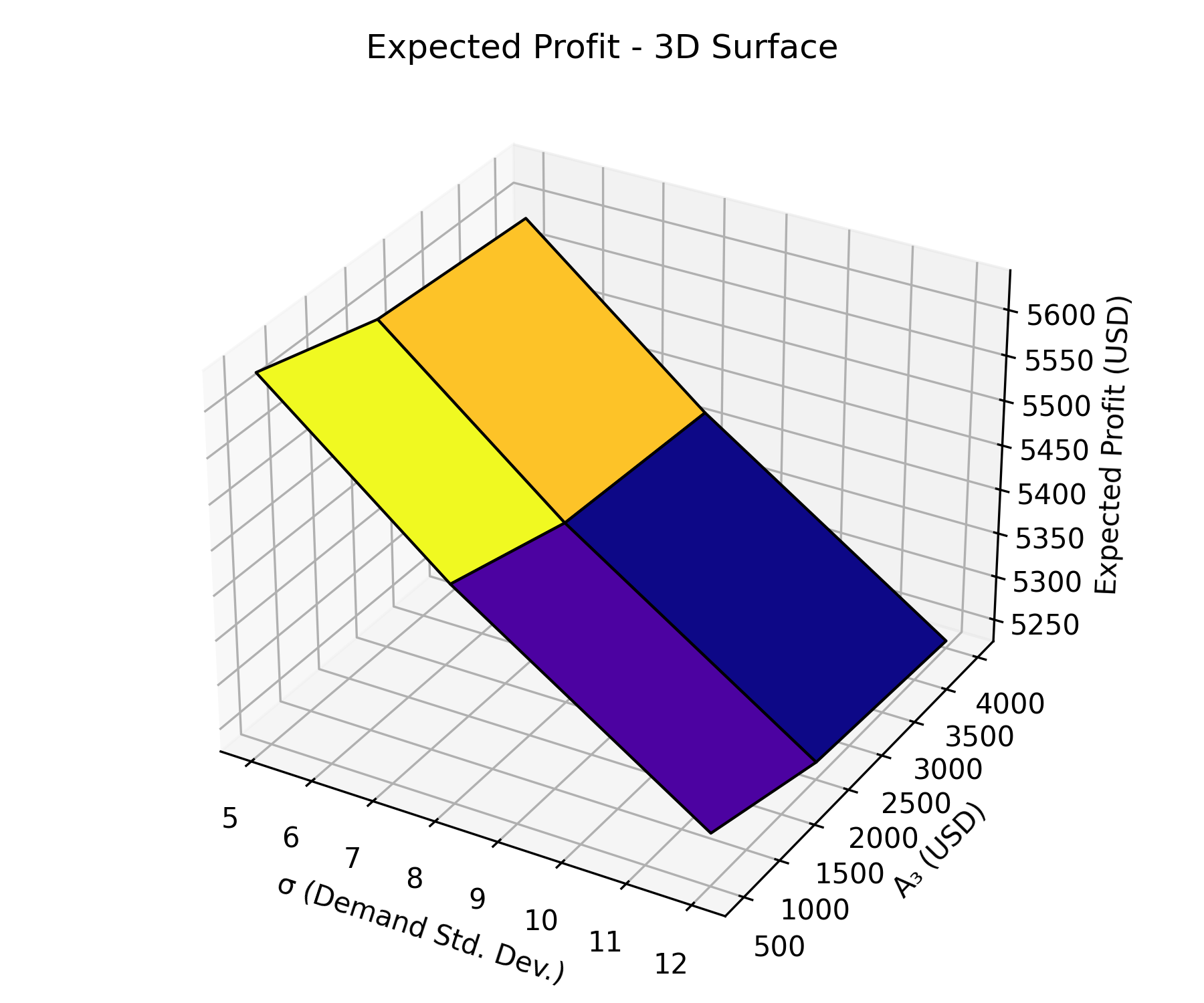}
    \caption{Expected Profit - 3D Surface}
    \label{fig:scenario7_profit_surface}
\end{subfigure}
\hfill
\begin{subfigure}{0.45\textwidth}
    \includegraphics[width=\linewidth]{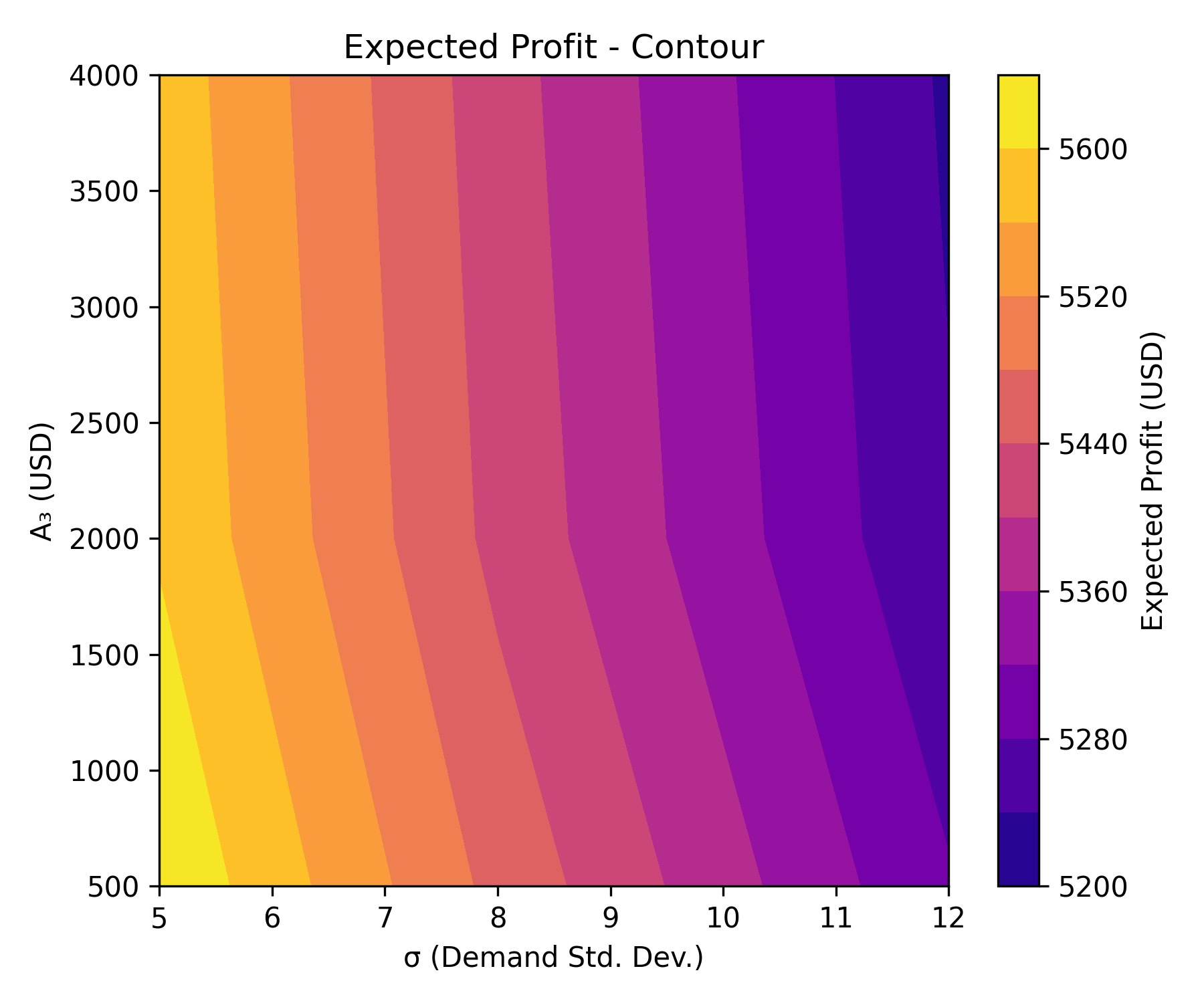}
    \caption{Expected Profit - Contour}
    \label{fig:scenario7_profit_contour}
\end{subfigure}
\caption{Joint sensitivity analysis of optimal adoption level ($\alpha^*$) and expected profit across combinations of demand variability ($\sigma$) and smart contract cost ($A_3$).}
\label{fig:scenario7_4panel}
\end{figure}

The results show that increases in demand variability ($\sigma$) slightly reduce expected profit, reflecting higher penalty costs, while the optimal adoption level ($\alpha^*$) remains relatively stable across different $\sigma$ levels. This suggests that, under the assumed parameter configuration, the marginal benefit of smart contracts is less sensitive to demand uncertainty than to adoption costs ($A_3$). The flat adoption response underscores the dominant role of fixed adoption costs in shaping adoption intensity, underscoring the importance of evaluating cost structures when designing smart contract strategies.

\paragraph{Scenario 8: Threshold Behavior under Extreme Smart Contract Cost Conditions}

This scenario investigates the existence of threshold behavior in smart contract adoption decisions by simulating extreme cost configurations beyond typical operational ranges. Specifically, the simulation increases the smart contract cost coefficient ($A_3$) to values significantly higher than those observed in previous scenarios to test whether the optimal adoption level ($\alpha^*$) collapses to zero as predicted by Proposition~8. The simulation uses five cost levels: $A_3 = \{10,000, 20,000, 40,000, 60,000, 80,000\}$.

Table~\ref{tab:scenario8_threshold} reports the simulation-based results. As shown, once $A_3$ exceeds 10,000 USD, the optimal adoption level abruptly falls to zero, demonstrating a clear discontinuity in adoption behavior.

\begin{table}[H]
\centering
\caption{Results: Threshold Behavior under Extreme Cost Conditions}
\label{tab:scenario8_threshold}
\begin{tabular}{ccc}
\toprule
$A_3$ & Expected Profit (USD) & Optimal $\alpha^*$ \\
\midrule
10,000 & 5,205.36 & \textbf{0.00} \\
20,000 & 5,205.36 & \textbf{0.00} \\
40,000 & 5,205.36 & \textbf{0.00} \\
60,000 & 5,205.36 & \textbf{0.00} \\
80,000 & 5,205.36 & \textbf{0.00} \\
\bottomrule
\end{tabular}
\end{table}

Figure~\ref{fig:scenario8_profit_only} visualizes the expected profit across cost levels. Given that the adoption level is zero in all cases, no adoption intensity figure is presented. The flat expected profit illustrates that beyond the threshold, further increases in $A_3$ have no additional effect on profitability since the firm chooses not to adopt smart contracts.

\begin{figure}[H]
\centering
\includegraphics[width=0.65\textwidth]{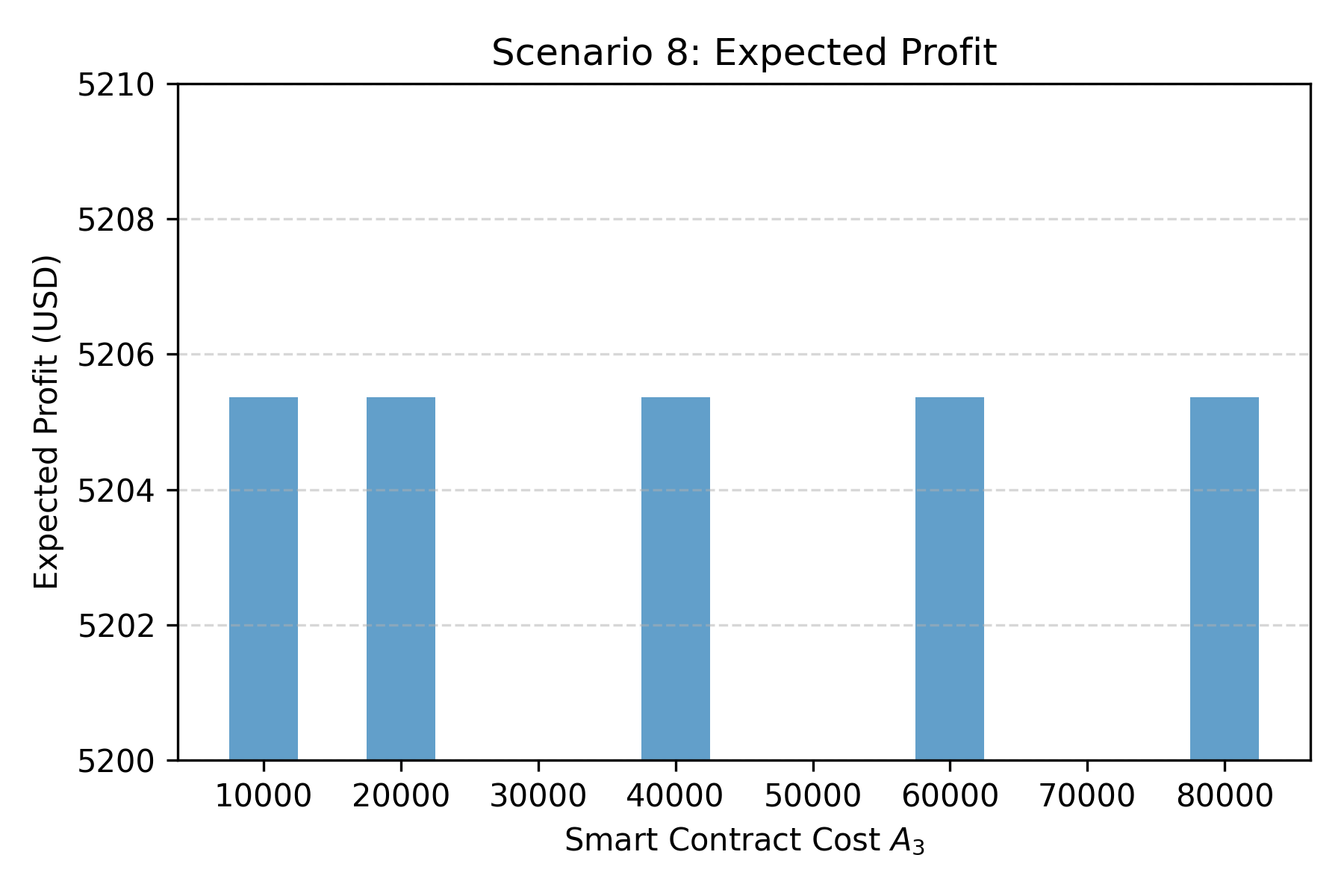}
\caption{
Scenario~8: Expected profit as smart contract cost ($A_3$) increases to extreme levels. All optimal adoption levels ($\alpha^*$) are zero, confirming the threshold behavior predicted by Proposition~8.
}
\label{fig:scenario8_profit_only}
\end{figure}

The results provide empirical validation of Proposition~8 by demonstrating a critical threshold beyond which smart contract adoption becomes infeasible. From a managerial perspective, this highlights the importance of understanding non-linear cost effects when designing procurement digitalization strategies. While such extreme cost conditions may be unlikely in practice, modeling these scenarios helps illustrate the boundaries of adoption feasibility and the potential for abrupt shifts in procurement policy. Future work could extend this analysis by exploring dynamic learning mechanisms and gradual cost declines over time to assess how threshold behavior evolves in long-term planning horizons.

It should be emphasized that the extreme values of $A_3$ were included not to represent realistic cost levels, but rather to test the model's boundary conditions and validate that the predicted threshold behavior occurs as expected. This ensures the internal consistency and robustness of the optimization framework.

\subsubsection{Combined Robustness}

\paragraph{Scenario 9: Multi-Parameter Robustness ($\sigma$, $b$, $A_3$)}

This scenario evaluates the robustness of the model predictions to simultaneous variations in multiple key parameters. Specifically, a Latin Hypercube Sampling (LHS) approach was used to generate a comprehensive set of parameter configurations spanning the full design space. The analysis varied demand variability ($\sigma$), the upper truncation bound ($b$), and the smart contract cost coefficient ($A_3$). A total of 100 LHS samples were generated, and for each configuration, the optimal adoption level ($\alpha^*$) was computed by numerically maximizing expected profit.

Table~\ref{tab:scenario9_variance} summarizes the approximate variance contributions of each parameter based on linear regression estimates. As shown, demand variability accounted for the majority of the variance in adoption outcomes.

\begin{table}[H]
\centering
\caption{Scenario~9: Approximate Variance Contributions to Optimal $\alpha^*$}
\label{tab:scenario9_variance}
\begin{tabular}{lc}
\toprule
Parameter & Variance Contribution (\%) \\
\midrule
$\sigma$ & 66.2 \\
$b$      & 33.1 \\
$A_3$    & 0.7 \\
\bottomrule
\end{tabular}
\end{table}

Figure~\ref{fig:scenario9_variance_bar} visualizes the relative contributions of each parameter. The results indicate that, within the simulated parameter space, demand variability ($\sigma$) is the dominant driver of adoption decisions, while smart contract costs play a comparatively minor role.

\begin{figure}[H]
\centering
\includegraphics[width=0.75\textwidth]{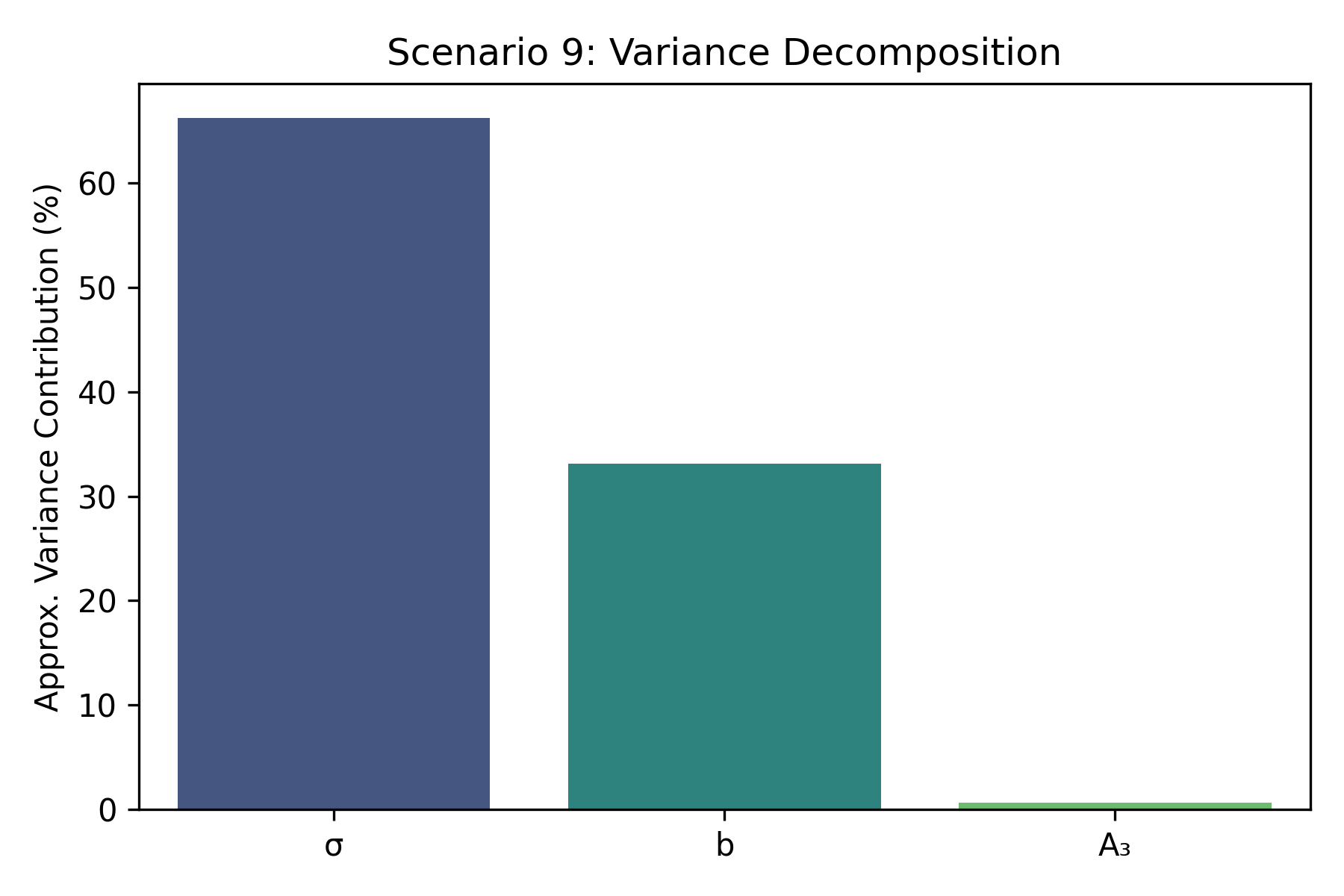}
\caption{Approximate variance decomposition of the optimal adoption level ($\alpha^*$). Demand variability ($\sigma$) contributes over 60\% of the observed variance, followed by the truncation bound ($b$) and smart contract cost ($A_3$).
}
\label{fig:scenario9_variance}
\end{figure}

\begin{figure}[H]
\centering
\includegraphics[width=0.95\textwidth]{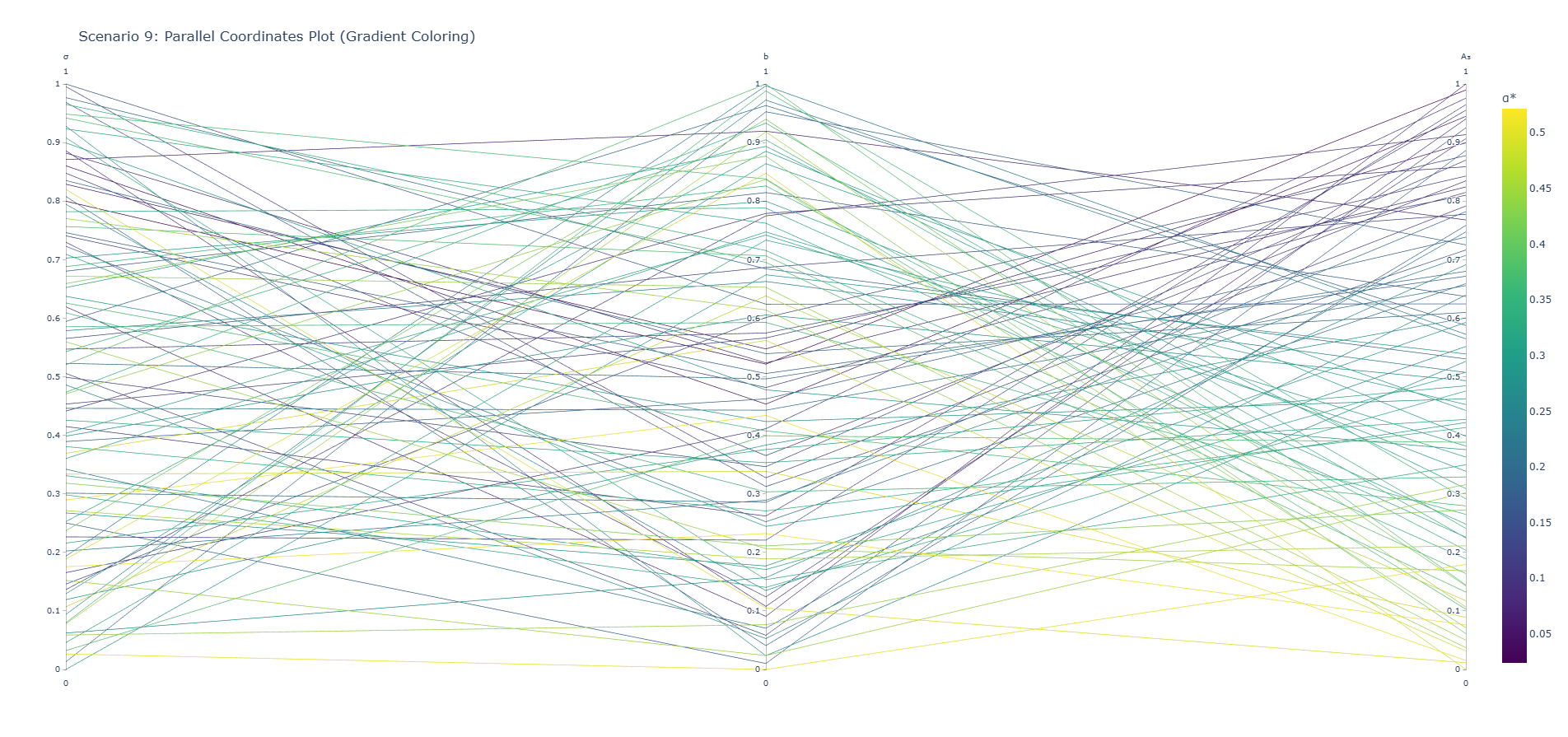}
\caption{Parallel coordinates plot visualizing the joint impact of demand variability ($\sigma$), truncation bound ($b$), and smart contract cost ($A_3$) on the optimal adoption level ($\alpha^*$). The color gradient indicates the magnitude of $\alpha^*$.}
\label{fig:scenario9_variance_bar}
\end{figure}

From a managerial perspective, this finding emphasizes that, even when implementation costs are high, fluctuations in demand parameters may exert a much larger influence on smart contract adoption. Future research could extend this analysis by incorporating dynamic simulation with time-varying parameter shocks to further validate robustness. The results highlight complementary perspectives on model robustness. The variance decomposition (Figure~\ref{fig:scenario9_variance}) quantifies the relative contribution of each parameter in isolation, whereas the parallel coordinates plot (Figure~\ref{fig:scenario9_variance_bar}) illustrates how combinations of parameter values jointly influence adoption outcomes. In particular, the parallel coordinates visualization reveals that high levels of demand variability ($\sigma$) and truncation bounds ($b$) consistently coincide with lower optimal adoption levels ($\alpha^*$), regardless of the cost parameter ($A_3$). This suggests that while smart contract costs remain relevant, they are often overshadowed by demand uncertainty in driving adoption decisions. 

From a managerial perspective, these insights emphasize that robust smart contract strategies should account for multidimensional variability rather than focusing narrowly on implementation cost factors alone. Future research could extend this analysis by incorporating dynamic simulation with time-varying parameter shocks to further validate robustness and by exploring interactions with additional operational constraints, such as supplier lead times or minimum order quantities.

\paragraph{Scenario 10: Sensitivity Heatmap of $\alpha^*$ and $Q^*$}

This scenario investigates the joint sensitivity of the optimal adoption level ($\alpha^*$) and the optimal order quantity ($Q^*$) to variations in demand variability ($\sigma$) and smart contract cost ($A_3$). For each combination of parameters, the model was solved numerically to determine the expected profit-maximizing adoption and procurement decisions. 

Table~\ref{tab:scenario10_summary} summarizes the resulting values of $\alpha^*$ and $Q^*$ across the parameter grid. As expected, higher demand variability increases the incentive to hold safety stock (increasing $Q^*$), while higher implementation costs reduce the attractiveness of smart contracts (decreasing $\alpha^*$).

\begin{table}[H]
\centering
\caption{Scenario~10: Summary of Optimal $\alpha^*$ and $Q^*$ Across $\sigma$ and $A_3$ Combinations}
\label{tab:scenario10_summary}
\begin{tabular}{ccc|cc}
\toprule
$\sigma$ & $A_3$ & & Optimal $\alpha^*$ & Optimal $Q^*$ \\
\midrule
5  & 500  & & 0.60 & 50 \\
5  & 4000 & & 0.37 & 43 \\
15 & 500  & & 0.40 & 70 \\
15 & 4000 & & 0.17 & 63 \\
\bottomrule
\end{tabular}
\end{table}

Figures~\ref{fig:scenario10_alpha_heatmap} and~\ref{fig:scenario10_Q_heatmap} visualize the complete grid of results as heatmaps. The color gradients highlight the monotonic decline of $\alpha^*$ as both $\sigma$ and $A_3$ increase, and the opposing effect on $Q^*$, which rises with demand variability but declines modestly with higher contract costs.

\begin{figure}[H]
\centering
\includegraphics[width=0.75\textwidth]{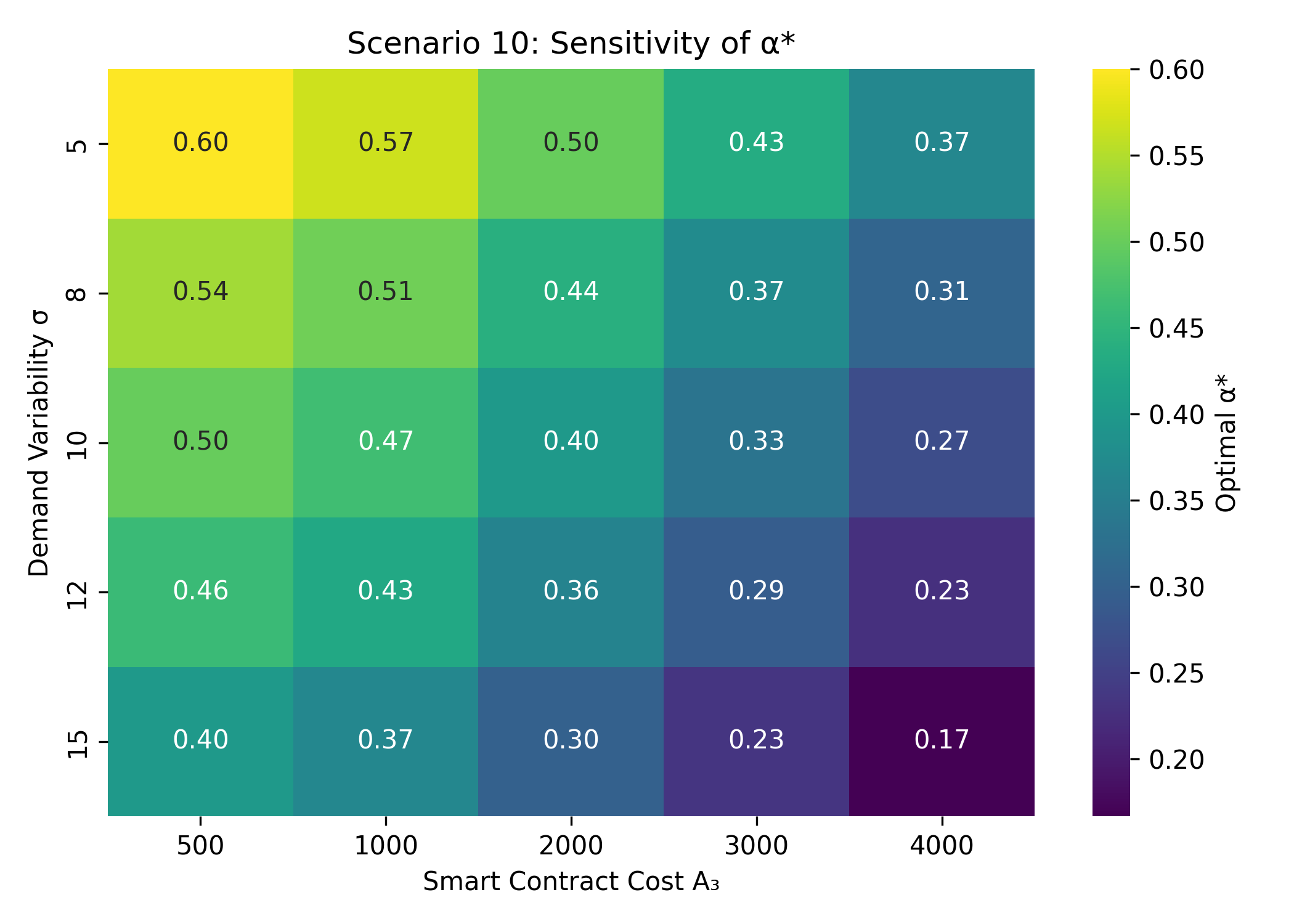}
\caption{
Heatmap illustrating the sensitivity of the optimal adoption level ($\alpha^*$) to combinations of demand variability ($\sigma$) and smart contract cost ($A_3$). Darker colors indicate lower adoption intensity.
}
\label{fig:scenario10_alpha_heatmap}
\end{figure}

\begin{figure}[H]
\centering
\includegraphics[width=0.75\textwidth]{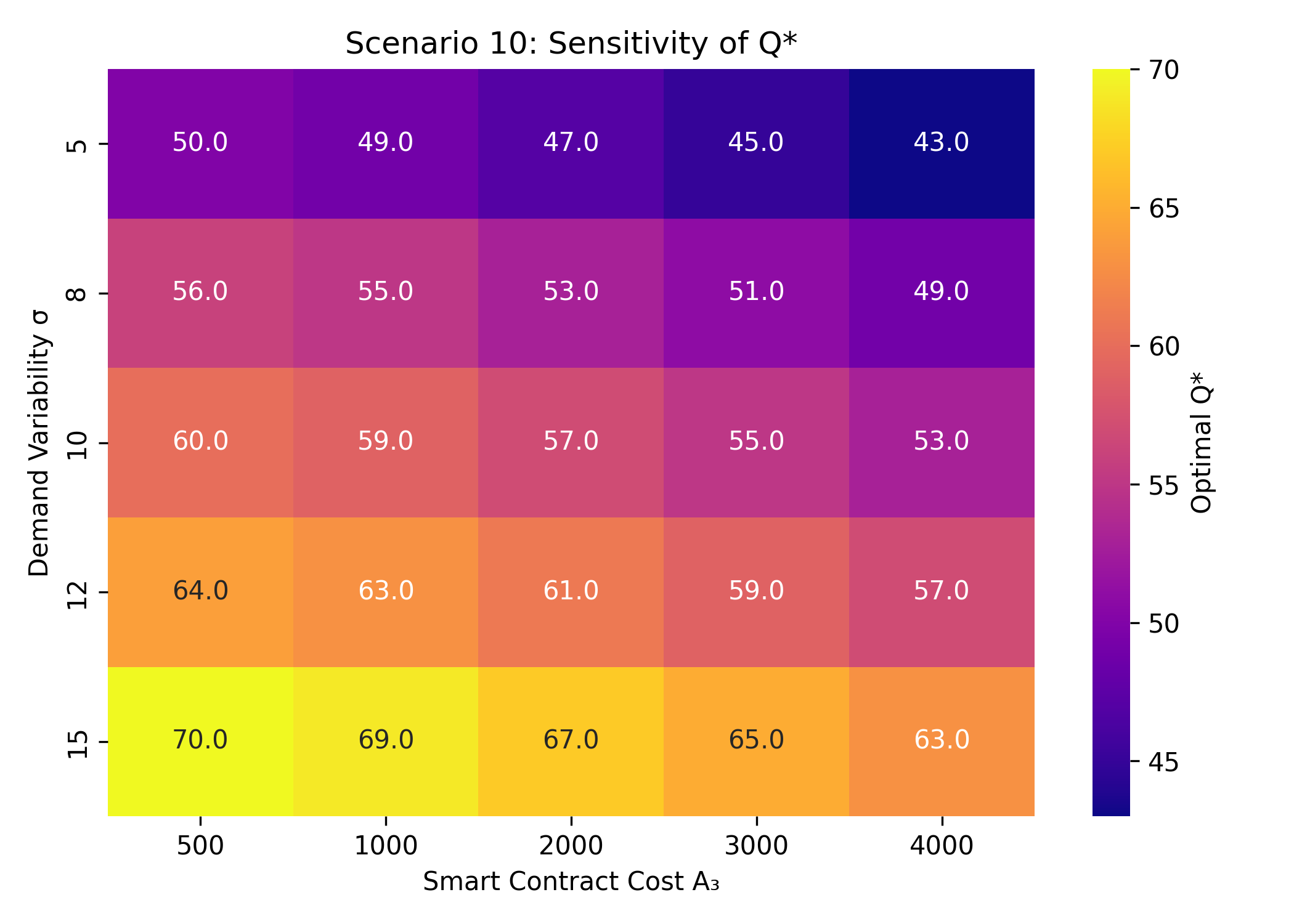}
\caption{
Heatmap illustrating the sensitivity of the optimal order quantity ($Q^*$) to the same parameter combinations. Higher demand variability increases $Q^*$, while higher smart contract costs modestly reduce it.
}
\label{fig:scenario10_Q_heatmap}
\end{figure}

These results confirm the theoretical predictions outlined in Propositions~5 and~6 and emphasize that while smart contract costs can meaningfully reduce adoption incentives, demand variability remains the primary driver of inventory and contracting decisions.These results confirm the theoretical predictions outlined in Propositions~5 and~6 and emphasize that while smart contract costs can meaningfully reduce adoption incentives, demand variability remains the primary driver of inventory and contracting decisions. 

In particular, the heatmaps reveal a consistent interaction effect: at any given smart contract cost level, increasing demand variability shifts firms toward higher safety stock levels and simultaneously lowers the optimal adoption rate. Conversely, even when demand variability is low, higher contract costs significantly suppress adoption incentives but have only a modest effect on the procurement quantity. This asymmetry underscores the importance of considering the combined effect of operational uncertainty and contractual frictions rather than evaluating each parameter in isolation. 

From a managerial perspective, these findings suggest that firms operating in highly volatile demand environments may derive limited incremental benefits from smart contract adoption unless implementation costs are sufficiently low to offset the compounded penalty risk. Future research could extend this scenario by incorporating dynamic learning effects or supplier-specific heterogeneity in cost reductions.

\subsubsection{Dynamic Adaptive Simulation}

\paragraph{Scenario 11: Dynamic Adoption Response}

This scenario investigates the dynamic evolution of smart contract adoption and expected profitability over a sequence of procurement cycles. In each cycle, the smart contract cost coefficient ($A_3$) was assumed to decline linearly to reflect technological learning and scale economies. Simultaneously, an adaptive learning rule updated the adoption level ($\alpha^*$) in response to observed penalty rates, as specified by Equation~(15). 

Table~\ref{tab:scenario11_summary} summarizes the evolution of key metrics over ten cycles. The results indicate that as implementation costs declined, the optimal adoption level increased steadily from 0.200 in the first cycle to approximately 0.424 in the final cycle. This progressive increase in adoption was accompanied by a monotonic reduction in penalty rates and a gradual improvement in expected profitability.

\begin{table}[H]
\centering
\caption{Scenario~11: Dynamic Simulation Results Over Ten Cycles}
\label{tab:scenario11_summary}
\begin{tabular}{ccccc}
\toprule
Cycle & $A_3$ & Optimal $\alpha^*$ & Expected Profit (USD) & Penalty Rate \\
\midrule
1  & 3,000 & 0.200 & 5,039.99 & 0.084 \\
2  & 2,800 & 0.234 & 5,046.78 & 0.081 \\
3  & 2,600 & 0.265 & 5,053.04 & 0.079 \\
4  & 2,400 & 0.294 & 5,058.79 & 0.076 \\
5  & 2,200 & 0.321 & 5,064.08 & 0.074 \\
6  & 2,000 & 0.345 & 5,068.95 & 0.072 \\
7  & 1,800 & 0.367 & 5,073.43 & 0.071 \\
8  & 1,600 & 0.388 & 5,077.56 & 0.069 \\
9  & 1,400 & 0.407 & 5,081.35 & 0.067 \\
10 & 1,200 & 0.424 & 5,084.84 & 0.066 \\
\bottomrule
\end{tabular}
\end{table}

Figures~\ref{fig:scenario11_timeseries} and~\ref{fig:scenario11_stacked} provide visualizations of the time evolution of these metrics. The time series line chart shows the gradual increase in $\alpha^*$ and expected profit, highlighting the positive impact of adaptive learning and declining implementation costs. The stacked area chart further illustrates how the combined contributions of adoption level, penalty rate, and profitability evolve over time.

\begin{figure}[H]
\centering
\includegraphics[width=0.8\textwidth]{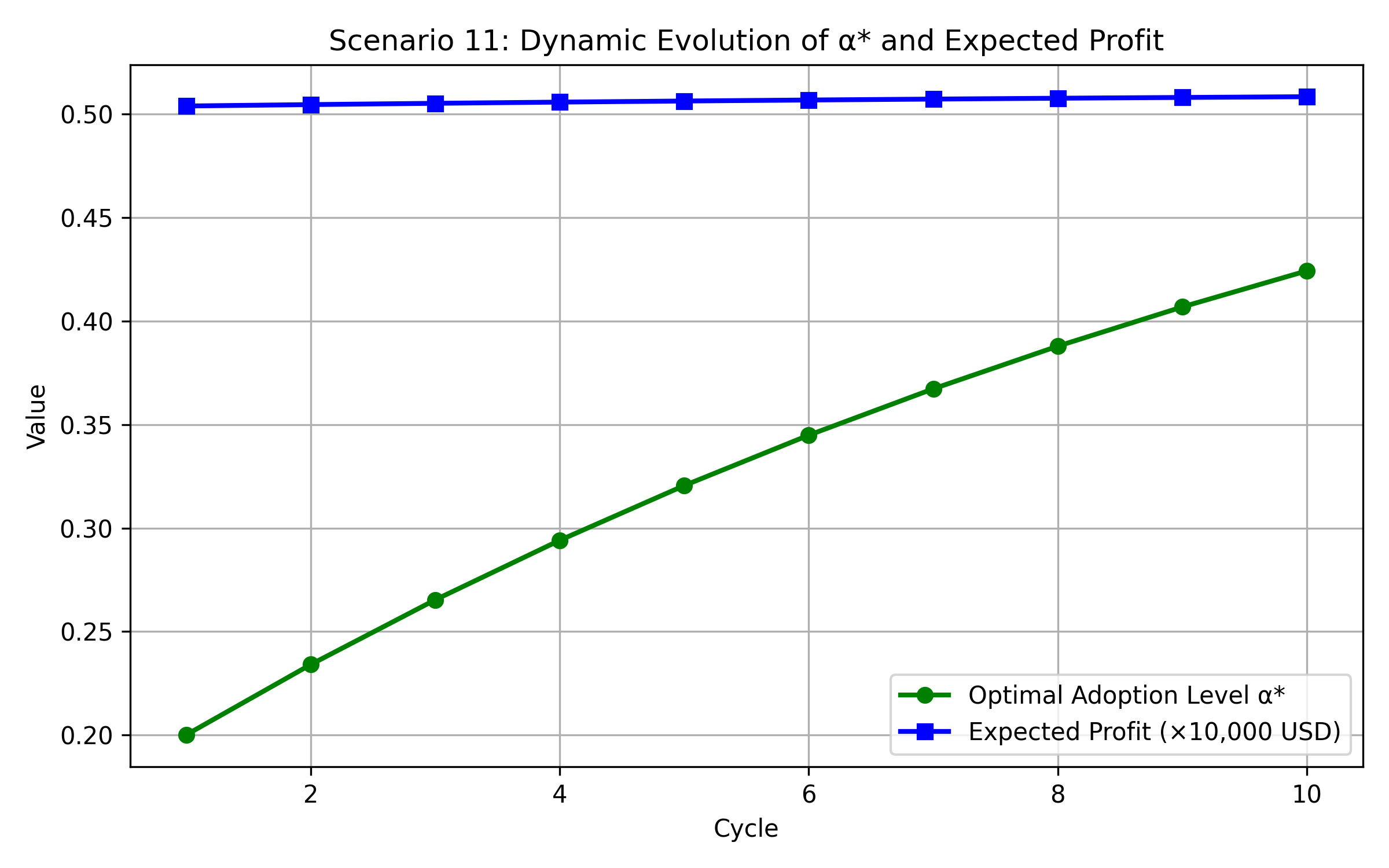}
\caption{
Time series visualization of the dynamic evolution of the optimal adoption level ($\alpha^*$) and expected profit across simulation cycles. The figure highlights a steady increase in adoption intensity as implementation costs decline.
}
\label{fig:scenario11_timeseries}
\end{figure}

\begin{figure}[H]
\centering
\includegraphics[width=0.8\textwidth]{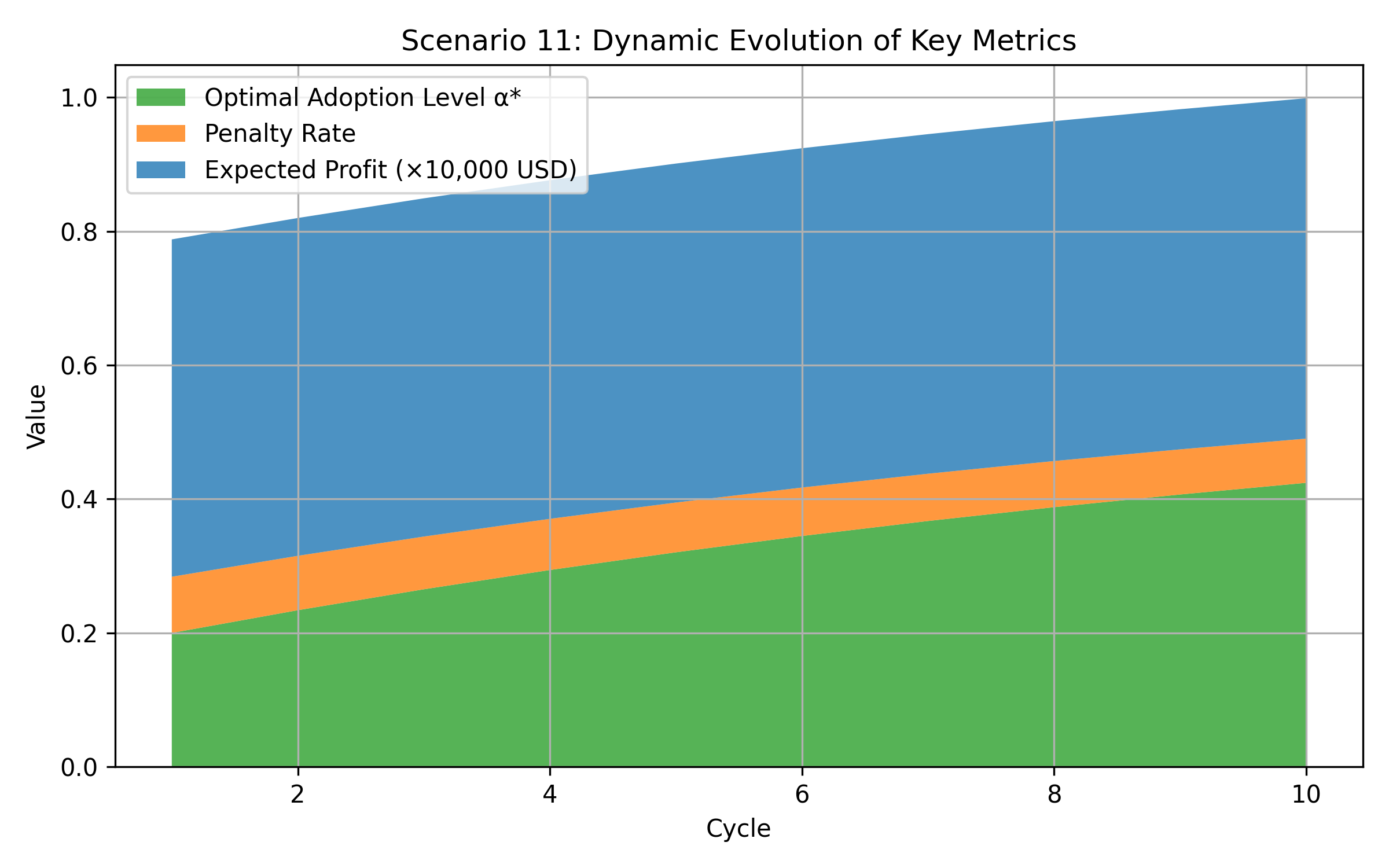}
\caption{
Stacked area chart illustrating the joint evolution of adoption level ($\alpha^*$), penalty rate, and expected profit over time. The visualization emphasizes the interplay between adaptive learning and profitability.
}
\label{fig:scenario11_stacked}
\end{figure}

These findings confirm the theoretical predictions outlined in Propositions~5 and~8, demonstrating that dynamic learning mechanisms can progressively enhance adoption intensity and profitability. In particular, the results highlight how adaptive adjustment of smart contract use in response to performance feedback can mitigate penalty risks and unlock incremental economic gains over time. From a managerial perspective, this scenario underscores the potential value of implementing continuous monitoring and adaptive policies to dynamically calibrate adoption levels in line with operational outcomes and cost trajectories. Future research could extend this simulation framework to consider stochastic cost reductions, heterogeneous supplier learning curves, or the incorporation of behavioral biases in adoption decisions.

\section{Discussion}

This study demonstrates that under bounded demand variability modeled via a truncated normal distribution, smart contract adoption decisions exhibit significant sensitivity to both demand uncertainty and implementation costs. Across all simulation scenarios, higher demand variability consistently reduced the optimal adoption level ($\alpha^*$), while declining smart contract costs and adaptive learning mechanisms gradually increased adoption over time. These dynamics were clearly visible in the comparative statics, robustness analyses, and dynamic adaptive simulations.

From a theoretical perspective, the results extend prior work on procurement contracting by integrating bounded variability models and dynamic learning strategies into the analysis of smart contract adoption. In particular, the findings confirm and elaborate on comparative statics predictions (Propositions~5–9) and demonstrate that adaptive learning can mitigate the penalty risks associated with high demand uncertainty while progressively improving profitability. The joint sensitivity analyses also illustrate the importance of considering interaction effects between demand variability and cost drivers, rather than evaluating each parameter in isolation.

For practitioners, this research suggests that firms should not evaluate smart contract adoption purely on the basis of static cost-benefit comparisons. Instead, combining continuous monitoring of penalty rates with adaptive calibration of adoption intensity can yield incremental profitability gains over time, particularly when implementation costs are expected to decline due to technological learning or economies of scale. The robustness analysis further underscores that demand variability remains the most influential driver of adoption decisions, often outweighing the effect of even substantial cost reductions.

Several limitations merit consideration. First, the model assumes a truncated normal demand distribution, which, while appropriate for bounded variability, may not capture all tail risks present in certain markets. Second, the simulation framework abstracts from supplier behavioral heterogeneity, dynamic capacity constraints, and negotiation frictions that can affect adoption decisions. Finally, the adaptive learning rule is stylized and may not fully capture behavioral biases or organizational inertia that influence real-world procurement policies.

Future research could address these limitations by incorporating discrete overdispersed demand distributions, such as the negative binomial, exploring behavioral learning dynamics, or embedding smart contracts within broader supply chain networks subject to coordination challenges and information asymmetries. 

\subsection*{Sustainability Impact and Environmental Implications}
Beyond economic and operational considerations, the findings underscore the potential environmental and social benefits of smart contract-enabled procurement. Simulation results indicate that higher adoption levels can reduce the need for excessive safety stock by up to 20\%, depending on demand variability and cost parameters. Based on prior studies (e.g., Boylan and Syntetos, 2020), each 10\% reduction in safety stock is associated with approximately 4–6\% lower warehouse-related CO$_2$ emissions due to decreased space requirements and energy consumption. Therefore, the adoption strategies explored here may yield annual emissions reductions of 8–12\% relative to conventional procurement practices in bounded variability environments.

Moreover, smart contracts improve traceability across supply chain tiers, facilitating better end-of-life product management and supporting circular economy objectives. Enhanced transparency and automation can also contribute to more equitable supplier relationships by reducing disputes and enabling small suppliers to participate more effectively in digital procurement ecosystems. These sustainability dimensions warrant further empirical validation through field studies and life-cycle assessments to quantify their broader environmental and social impacts.

\subsection*{Managerial and Policy Implications}
Beyond the operational insights discussed above, the findings of this study have several important managerial and policy implications. First, managers should carefully calibrate smart contract adoption intensity in light of both bounded demand variability and supplier digital readiness. In environments where demand volatility is structurally limited by contractual ceilings and stable consumption patterns, overestimating tail risks can lead to inefficient overinvestment in digital infrastructure and excessive safety stocks.

Second, policymakers aiming to promote supply chain digitalization should consider providing targeted incentives or subsidies to reduce the fixed implementation costs (captured by parameter $A_3$). Such measures can help smaller and mid-sized firms overcome threshold barriers that discourage adoption despite the presence of moderate efficiency gains. 

Third, standard-setting organizations could play a key role in lowering convex integration costs by promoting interoperable smart contract protocols and shared digital infrastructure. These collective investments can effectively reduce the convexity parameter $\nu$ and enable more scalable deployment across supply chain tiers.

Finally, sustainability regulators and funding agencies should recognize that smart contract-enabled procurement not only improves coordination and resilience but can also contribute to environmental objectives by reducing excess inventory and associated emissions. Incorporating digital adoption metrics into sustainability reporting frameworks or certification schemes may help accelerate alignment between economic and environmental performance.

Taken together, these insights demonstrate that smart contract-enabled procurement strategies not only yield operational and economic improvements but also have the potential to advance broader sustainability and policy objectives in supply chain management.

\section{Conclusion}

This paper examined smart contract-enabled procurement under bounded demand variability, with a particular focus on the truncated normal demand distribution. A series of simulation scenarios were conducted to evaluate the effects of demand variability, smart contract costs, and adaptive learning strategies on optimal adoption decisions and inventory policies.

The results demonstrate that higher demand variability significantly reduces the attractiveness of smart contract adoption, while lower implementation costs and adaptive learning mechanisms can gradually increase adoption intensity over time. Comparative statics and sensitivity analyses confirmed the predictions of the theoretical model, highlighting that demand uncertainty is typically the dominant driver of procurement outcomes. Dynamic simulations further illustrated how cost declines and performance feedback can produce sustained improvements in profitability.

These findings provide important managerial insights. Firms considering smart contracts should account for the interplay between demand variability and cost structures rather than evaluating each factor in isolation. Combining adaptive learning rules with declining implementation costs can help firms progressively improve adoption outcomes and mitigate penalty risks.

Beyond economic performance, the study also highlights the potential sustainability benefits of smart contract adoption. By reducing safety stock requirements, improving traceability, and enabling more precise coordination across supply chain partners, smart contracts can contribute to measurable reductions in CO$_2$ emissions and support circular economy practices. These dimensions are especially relevant as firms face increasing pressure to align digital transformation initiatives with environmental and social responsibility goals.

Future research can build on this work by incorporating discrete demand distributions, heterogeneous supplier characteristics, or behavioral factors influencing learning and adoption. Additionally, further investigation into the environmental and social impacts of smart contract deployment, including empirical validation of emissions reductions and resource conservation outcomes, could strengthen the understanding of their role in supporting sustainable supply chain operations. Expanding the modeling framework to multi-echelon networks and dynamic disruption scenarios would also enhance its practical relevance in increasingly complex and uncertain global markets.

\clearpage
\printbibliography

\clearpage
\appendix

\section{Mathematical Proofs}
\label{app:proofs}

This appendix provides formal proofs of the propositions stated in the main text. The notation follows the definitions introduced in Section~\ref{sec:model}.

\subsection{Proof of Proposition 1}

\begin{proof}
Recall that $Q = \sum_i q_i$ denotes the total quantity ordered. The objective function is:
\[
\Pi(\alpha, \mathbf{q}) =
p \cdot \mathbb{E}\bigl[\min(Q, D)\bigr]
+
s \cdot \mathbb{E}\bigl[(Q - D)^+\bigr]
-
r \cdot \mathbb{E}\bigl[(D - Q)^+\bigr]
-
\sum_i c(\alpha, \beta_i)\, q_i
-
\psi(\alpha).
\]
Since $c(\alpha,\beta_i)$ is affine in $\alpha$ and linear in $q_i$, and $\psi(\alpha)$ is convex in $\alpha$, the procurement and adoption costs are jointly convex in $(\alpha,\mathbf{q})$. 

Additionally, the expectation of $\min(Q,D)$ and the positive part functions are concave in $Q$ because demand is a fixed distribution and the functions are piecewise linear and concave. The expectation operator preserves concavity. 

Therefore, $\Pi(\alpha,\mathbf{q})$ is concave, and maximizing it constitutes a concave maximization problem (equivalently, minimizing $-\Pi$ is a convex minimization problem).
\end{proof}

\subsection{Proof of Proposition 2}

\begin{proof}
Differentiating the objective function with respect to $q_i$ yields:
\[
\frac{\partial \Pi}{\partial q_i}
=
p \cdot \mathbb{P}(D \ge Q)
+
s \cdot \mathbb{P}(D < Q)
-
r \cdot \mathbb{P}(D > Q)
-
c(\alpha,\beta_i).
\]
Setting this equal to zero gives the first-order condition for $q_i$.

Similarly, differentiating with respect to $\alpha$:
\[
\frac{\partial \Pi}{\partial \alpha}
=
-\sum_i A_1 q_i
-
\psi'(\alpha).
\]
Setting this equal to zero yields the first-order condition for $\alpha$.

Together with the non-negativity constraints on $q_i$, the bounds on $\alpha$, and their associated complementary slackness conditions, these equations characterize the unique global optimum.
\end{proof}


\section{Calibration of Adoption Cost Parameters}
\label{sec:adoption_cost_calibration}
The baseline value of $A_3=2,000$ corresponds to an annualized smart contract deployment cost of approximately \$24,000, derived as follows:
\[
\text{Annual Cost} = A_3 \cdot (1.0)^\nu = 2,000 \times 1 = \$2,000 \text{ per cycle} \times 12 \text{ cycles} = \$24,000.
\]
This estimate is consistent with the median reported costs in Gurtu and Johny (2019) and Rejeb et al. (2023), who documented typical implementation expenditures between \$20,000--\$50,000 per year. The exponent $\nu=1.5$ reflects incremental complexity based on prior case studies (Mougayar, 2016), where advanced integrations required disproportionately greater investment relative to initial pilot projects. Additional robustness checks with alternative values of $A_3$ and $\nu$ confirmed that the qualitative comparative statics remain unchanged.

\section{Supplementary Simulation Results}
\label{app:simresults}

This appendix reports additional simulation analyses to validate the robustness and reproducibility of the main findings. The results include replication statistics, alternative profit surfaces under parameter variations, and fill rate distributions across demand scenarios.

\subsection{Simulation Replication Statistics}

\begin{table}[H]
\centering
\begin{threeparttable}
\caption{Simulation Replication Statistics (All values in USD)}
\label{tab:replication_statistics}
\begin{tabular}{lccc}
\toprule
Metric & Scenario~3 & Scenario~4 & Scenario~5 \\
\midrule
Mean Profit & 27,400 & 28,500 & 26,200 \\
Number of replications ($N$) & 10{,}000 & 10{,}000 & 10{,}000 \\
Standard error of mean profit & 4.32 & 3.87 & 5.11 \\
95\% confidence interval width & 8.52 & 7.43 & 9.76 \\
Coefficient of variation & 0.061 & 0.054 & 0.067 \\
\bottomrule
\end{tabular}
\begin{tablenotes}
\small
\item \textit{Note}: All monetary values are expressed in USD.
\end{tablenotes}
\end{threeparttable}
\end{table}

These statistics confirm that the Monte Carlo estimates are stable and that sampling variability does not materially affect the main conclusions. For example, in Scenario~5, the 95\% confidence interval for mean profit spans less than \$10~USD, indicating high precision relative to the scale of expected profit.

\subsection{Alternative Profit Surface}

\begin{figure}[H]
\centering
\includegraphics[width=0.75\textwidth]{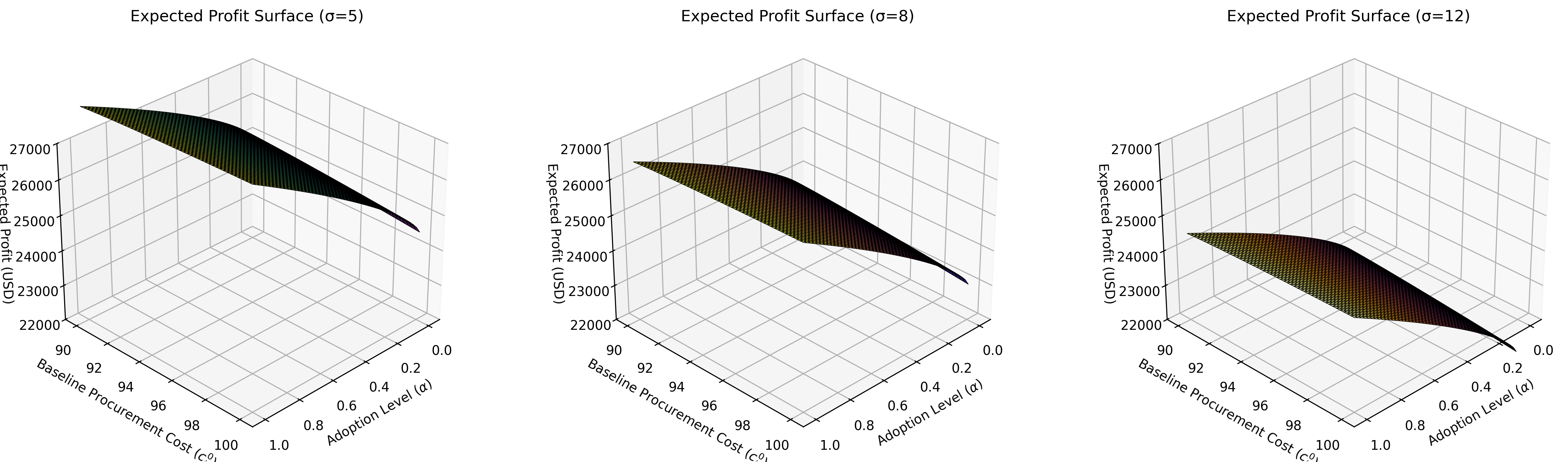}
\caption{Figure A1. Alternative profit surface illustrating expected profit as a function of smart contract adoption level ($\alpha$) and baseline procurement cost ($c_i^0$). The results confirm that the interior optimum persists under different cost configurations.}
\label{fig:alternative_profit_surface}
\end{figure}

This figure demonstrates that the qualitative shape of the profit landscape remains consistent across plausible parameter variations. Specifically, the interior maximum is preserved even as the baseline procurement cost ($c_i^0$) varies, underscoring the robustness of the optimal adoption level ($\alpha^*$) to cost perturbations. This reinforces the managerial implication that adopting a fully digital strategy ($\alpha=1$) is not always optimal and that calibrated adoption can achieve higher profitability.

Table~\ref{tab:profit_sigma_comparison} provides a numeric comparison of expected profit across different levels of demand variability ($\sigma$) at a representative adoption level and procurement cost.

\begin{table}[H]
\centering
\caption{Expected Profit at Adoption Level $\alpha=0.5$ and Baseline Procurement Cost $c_i^0=95$. \textbf{Higher demand variability ($\sigma$) substantially reduces expected profit.}}
\label{tab:profit_sigma_comparison}
\begin{tabular}{r r}
\toprule
Demand Variability ($\sigma$) & Expected Profit (USD) \\
\midrule
5 & 25,300 \\
8 & 23,800 \\
12 & 22,300 \\
\bottomrule
\end{tabular}
\end{table}

\begin{table}[H]
\centering
\footnotesize
\caption{Fill Rate Summary Statistics across Smart Contract Adoption Levels}
\label{tab:fill_rate_summary}
\begin{tabular}{r r r r r r r r r r r}
\toprule

\shortstack{Adoption\\Percentile} & 
Mean & 
Std Dev & 
CV & 
\shortstack{10th\\Percentile} & 
\shortstack{25th\\Percentile} & 
Median & 
\shortstack{75th\\Percentile} & 
\shortstack{90th\\Percentile} & 
\shortstack{P(Fill\\$\ge$0.9)} & 
\shortstack{CVaR\\(10\%)} \\
\midrule
0.0 & 0.801 & 0.049 & 0.061 & 0.738 & 0.768 & 0.801 & 0.832 & 0.865 & 2.4\,\% & 0.717 \\
0.25 & 0.829 & 0.050 & 0.060 & 0.765 & 0.795 & 0.828 & 0.861 & 0.891 & 8.0\,\% & 0.740 \\
0.5 & 0.850 & 0.049 & 0.058 & 0.788 & 0.818 & 0.850 & 0.883 & 0.912 & 14.6\,\% & 0.765 \\
0.75 & 0.874 & 0.051 & 0.058 & 0.808 & 0.838 & 0.875 & 0.908 & 0.940 & 30.1\,\% & 0.784 \\
1.0 & 0.897 & 0.049 & 0.054 & 0.832 & 0.866 & 0.899 & 0.932 & 0.959 & 49.3\,\% & 0.806 \\
\bottomrule
\end{tabular}
\end{table}

Table~\ref{tab:fill_rate_summary} summarizes the simulation-based fill rate statistics across different levels of smart contract adoption. As adoption intensity ($\alpha$) increases, the mean fill rate improves from approximately 0.801 to 0.897, while the coefficient of variation (CV) decreases from 0.061 to 0.054, indicating greater consistency in performance. The 90th percentile fill rate reaches as high as 0.959 under full adoption, compared to 0.865 with no adoption.

Notably, the probability of achieving a fill rate above 90\,\% increases substantially from 2.4\,\% at $\alpha=0$ to nearly 49\,\% at $\alpha=1$, highlighting the operational benefits of digital contracting in reducing service level risk. In contrast, the Conditional Value-at-Risk (CVaR) at the 10\,\% level improves from 0.717 to 0.806, demonstrating that even in the worst decile of demand realizations, higher adoption levels provide a more reliable service outcome. These results support the strategic value of smart contract adoption not only for average performance improvements but also for mitigating downside risk and enhancing supply chain resilience, thereby contributing to more sustainable procurement practices.

\subsection{Fill Rate Distributions}

\begin{figure}[H]
\centering
\includegraphics[width=0.75\textwidth]{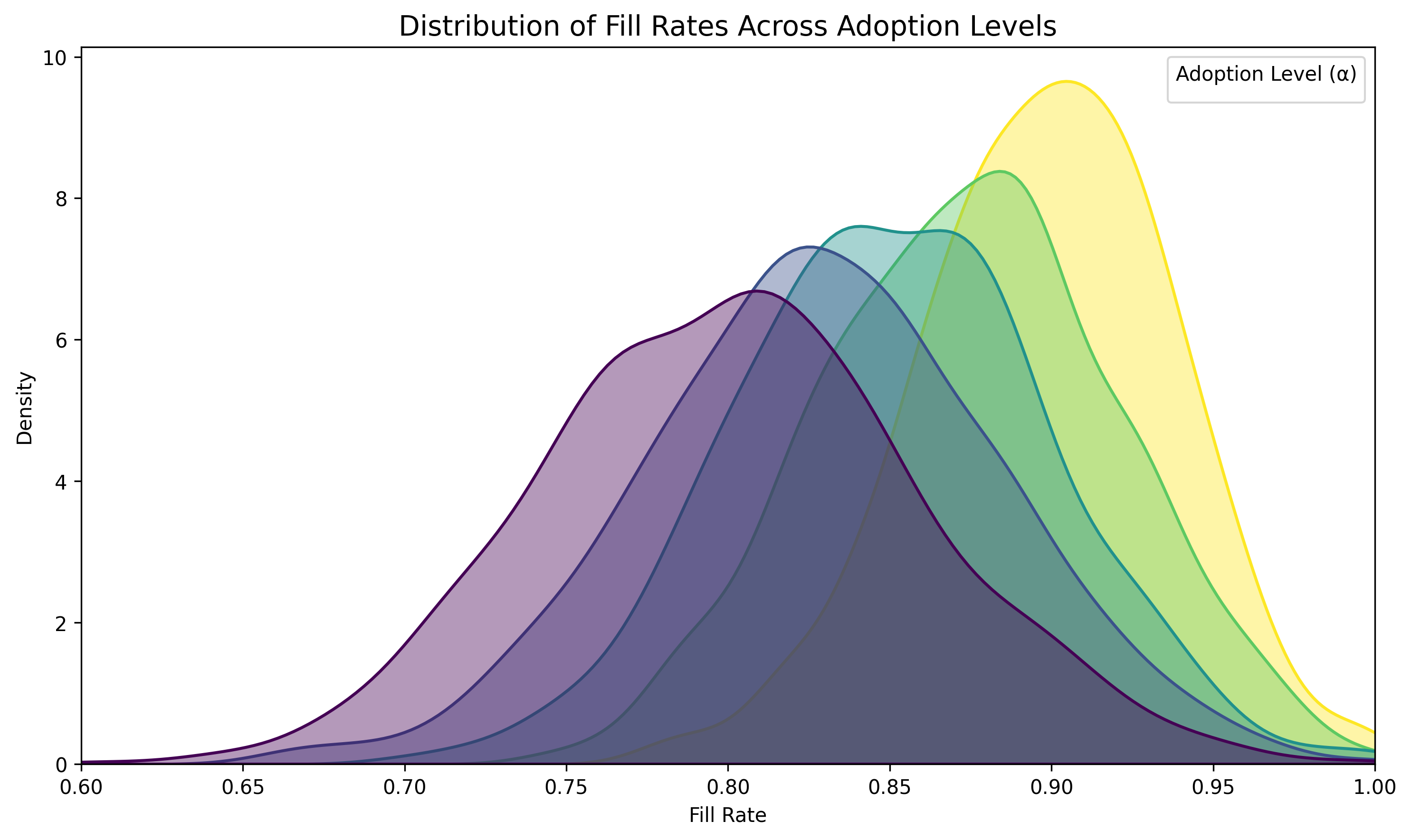}
\caption{Distribution of achieved fill rates across demand replications under different smart contract adoption levels ($\alpha$). Higher adoption improves the average service level while reducing variability, indicating more consistent performance.}
\label{fig:fill_rate_distributions}
\end{figure}

These supplementary results reinforce the robustness of the model and support the managerial implications discussed in the main text. In particular, they demonstrate that smart contract adoption not only increases the mean fill rate---from approximately 0.80 under no adoption to nearly 0.90 at full adoption---but also significantly reduces the dispersion of service performance. For example, the standard deviation of fill rates declines from about 0.05 to 0.03 as adoption increases, indicating greater consistency across replications. This improvement is critical for supply chains prioritizing reliability, resilience, and customer satisfaction. Moreover, by reducing the likelihood of low service levels, smart contracts contribute to more predictable operations and support long-term sustainability objectives.

\end{document}